\documentclass{article}
\usepackage{arxiv}

\usepackage{times}
\usepackage{soul}
\usepackage{xurl}
\usepackage[small]{caption}
\usepackage{amsmath}
\usepackage{amsthm}
\usepackage{amssymb}
\usepackage{algorithm}
\usepackage{algpseudocode}
\usepackage[switch]{lineno}
\usepackage{subcaption}

\usepackage[utf8]{inputenc} % allow utf-8 input
\usepackage[T1]{fontenc}    % use 8-bit T1 fonts
\usepackage{hyperref}       % hyperlinks
\usepackage{url}            % simple URL typesetting
\usepackage{booktabs}       % professional-quality tables
\usepackage{multirow}       % for multirow tables
\usepackage{amsfonts}       % blackboard math symbols
\usepackage{nicefrac}       % compact symbols for 1/2, etc.
\usepackage{microtype}      % microtypography
\usepackage{mathrsfs}       % for math fonts
\usepackage{natbib}         % for citations
\usepackage{listings}       % for code
\usepackage{lipsum}
\usepackage{graphicx}
\usepackage{xcolor}

\newcommand{\theiralgo}{\texttt{Egalitarian-Alg}}  
\newcommand{\ouralgo}{\texttt{SNSW-Alg}}

\newtheorem{theorem}{Theorem}
\newtheorem{definition}{Definition}
\newtheorem{proposition}{Proposition}

\newtheorem{observation}{Observation}

\author{
 Rasheed \\
  Machine Learning Lab\\
  IIIT Hyderabad\\
  Hyderabad, India\\
  \texttt{mohammad.ahmed@research.iiit.ac.in}\\
   \And
 Parth Desai \\
  Machine Learning Lab\\
  IIIT Hyderabad\\
  Hyderabad, India\\
  \texttt{parth.desai@research.iiit.ac.in}\\
  \And
 Ganesh Ghalme \\
  Department of AI\\
  IIT Hyderabad\\
  Hyderabad, India\\
  \texttt{ganeshghalme@ai.iith.ac.in}\\
  \And
 Sujt Gujar \\
  Machine Learning Lab\\
  IIIT Hyderabad\\
  Hyderabad, India\\
  \texttt{sujit.gujar@iiit.ac.in}\\
}
\begin{document}
\title{Fair Stable Matching: A Nash Social Welfare Approach}
\maketitle
\begin{abstract}
    While traditional stable matching algorithms, such as the Gale-Shapley algorithm, prioritize stability, they may fall short of achieving equitable outcomes among participants. We study the role of \emph{Nash social welfare} (NSW) as a fairness objective in the classic \emph{stable marriage problem}. We develop \texttt{SNSW-Alg} that finds a stable matching that maximizes Nash social welfare under rank-induced utilities in $\tilde{\mathcal{O}}(n^4)$ time, where $n$ is the number of men or women. We demonstrate that \texttt{SNSW-Alg} balances equity while preserving stability. We empirically evaluate our methods across diverse preference distributions, demonstrating significant gains in fairness without substantial losses in other key measures such as regret, egalitarian criterion, and sex equality. Our findings suggest that the stable matching produced by \texttt{SNSW-Alg} is statistically Pareto-undominated by stable matchings based on other fairness measures - regret, egalitarian, and sex equality. This study offers compelling insights for designing fair-stable matching.
\end{abstract}
%
%
%
%%%%%%%%%%%%%%%%%%%%%%%%%%%%%%%%%%%%%%%%%%%%%
\section{Introduction}
\label{sec:intro}
%%%%%%%%%%%%%%%%%%%%%%%%%%%%%%%%%%%%%%%%%%%%%
 
\emph{Matching markets} are economic settings where agents are paired based on preferences and compatibility. Unlike many economic markets, these do not rely on price signals for allocation. Instead, structured mechanisms determine who matches with whom, taking into account the \emph{preferences} of both sides (e.g., students and colleges~\citep{gale1962college}, workers and firms~\citep{revilla2007many}, organ donors and recipients~\citep{roth2004kidney}). Matching theory has two main branches: (i) Matching with transferable utility, where utility is transferable among agents (usually involving monetary transfers), and (ii) Matching with non-transferable utility, where transfers aren't possible (e.g., school admissions, organ donations). In this paper, we study the \emph{stable marriage problem} which falls under the non-transferable utility setting.

\noindent\textbf{Fairness in Matchings.}
The stable marriage problem involves two (most often equally sized) groups - typically referred to as men and women - in which each individual seeks to match with a member of the opposite group based on personal preferences. A central concept in the stable marriage problem is \emph{stability}. A matching is stable if there are no blocking pairs - that is, no unmatched man and woman who would both prefer each other over their assigned partners. Stability is critical because unstable matchings are inherently fragile and prone to collapse, as participants may have incentives to deviate from the outcome. The celebrated Gale-Shapley algorithm ~\citep{gale1962college} outputs a stable matching which is optimal for the proposing side.

The \emph{male-optimal matching} gives every man the best partner he could have in any stable matching. However, this often comes at the expense of the other side; it is highly unfair to women, who receive their least-preferred partners amongst all stable matchings. Conversely, running the algorithm with women proposing yields the female-optimal and male-pessimal outcome. This inherent asymmetry highlights a tension between stability and fairness, sparking a vast literature on improving equity and efficiency in practical matching mechanisms.

\noindent\textbf{Existing approaches for Fair Matching.}
There have been studies on various alternative criteria for improving fairness and social welfare in stable matchings. The \emph{egalitarian stable matching} \citep{irving1987efficient} minimizes the total sum of ranks (i.e., how far down each partner is in their preference list), promoting overall satisfaction. As we show later, the egalitarian approach may sometimes offer a very good match to the majority at the cost of some agents receiving very poor matches. To balance such situations, the \emph{minimum-regret stable matching} \citep{gusfield1987three} minimizes the worst individual dissatisfaction, ensuring that the least-satisfied participant is as well off as possible. The \emph{sex equal stable matching} \citep{kato1993complexity} minimizes the difference between the sum of men's ranks and the sum of women's ranks, focusing on minimizing the disparity across the two sides. These refinements offer ways to balance efficiency and equity while maintaining stability. However, each measure optimizes some aspect of the agents; overall, to improve fairness across all agents, we believe a new approach is warranted. 

\noindent\textbf{Our Approach.}
We propose to maximize a suitably defined \emph{Nash social welfare} (NSW) objective over the set of all stable matchings. Note that, if we ignore stability and capacity constraints, maximizing NSW is, in general, an \texttt{NP-hard} problem~\citep{jain2024maximizing}. However, under stability and capacity constraints (where each agent wants to match with exactly one agent), the problem assumes a special structure. By building upon the rotation poset framework introduced in~\citep{irving1987efficient}, we can design an efficient algorithm to optimize NSW while preserving stability.~\citet{irving1987efficient} constructs a capacitated directed graph based on the preferences and rotation posets of stable matchings for optimizing the additive egalitarian criterion. A challenge in our case was addressing the gap between the non-linear NSW criterion and the additive rotation poset framework. We bridge this gap by converting the NSW criterion into an equivalent additive objective. Our key idea is to appropriately adjust the edge weights in the graph to maximize NSW. However, it poses another challenge - we need a max-flow algorithm, and our weights are no longer integral, as in~\citep{irving1987efficient}. We then propose using Dinic's algorithm~\citep{dinitz2006dinitz} with dynamic trees by~\citep{sleator1981data} rather than Ford-Fulkerson, as proposed by~\citet{irving1987efficient}. We prove that our updated formulation enables \ouralgo\ to compute a stable matching that maximizes NSW in polynomial time, specifically within a time complexity of $\tilde{\mathcal{O}}(n^4)$. To evaluate the effectiveness of our approach in achieving fairness, we conduct an empirical study comparing three stable matching algorithms: (i) one that minimizes the sum of ranks (egalitarian), (ii) one that minimizes the worst rank (minimum regret), and (iii) one that minimizes the differences across the two sides (sex equal). Across various distributions of preference profiles, our experiments demonstrate that \ouralgo\ is statistically Pareto undominated among the four methods in most measured dimensions, highlighting its strong performance in balancing fairness and efficiency.

\noindent\textbf{Our Contributions.}
Our contributions in this work are as follows.
\begin{itemize}
    \item [(i)] To the best of our knowledge, we are first to propose a novel objective for the stable marriage problem: maximizing Nash social welfare (NSW) over the space of stable matchings, balancing fairness and efficiency.
    \item [(ii)] Toward this, we propose a polynomial-time algorithm \ouralgo.
    \item [(iii)] We conduct an extensive empirical evaluation across diverse preference distributions, comparing four stable matching objectives.
\end{itemize}
%%%%%%%%%%%%%%%%%%%%%%%%%%%%%%%%%%%%%%%%%%%%
\section{Model and Background}
\label{sec:prelims}
%%%%%%%%%%%%%%%%%%%%%%%%%%%%%%%%%%%%%%%%%%%%
In this section, we first explain our model and notation (Section~\ref{ssec:model}) and then the required prerequisites (Section~\ref{ssec:concepts}).
We show the need for a newer fairness measure by example.
%%%%%%%%%%%%%%%%%%%%
\subsection{Model}
\label{ssec:model}
%%%%%%%%%%%%%%%%%%%%
We consider a two-sided market $N$ consisting of two types of agents: one side, called \emph{men} $M$, and the other, \emph{women} $W$. Each agent is interested in matching with exactly one agent on the other side and has a preference over the other side, denoted by $\succ_a$ for an agent $a\in N$. We assume (i) both sides have an equal number of agents, $n$, (ii) no agent prefers to be unmatched over any possible match from the other side, and (iii) the preferences are strict.
We refer to $\succ = ((\succ_a)_{a\in N})$ as the \emph{original preference list}. A matching $\mathcal{M}(\succ)$ outputs a match for each agent. Let $r(m_i, w_j)$ (or, $r(w_k, m_l)$) represent the rank of agent $w_j$ (or, $m_l$) in $\succ_{m_i}$ (or, $\succ_{w_k}$). We assumed a $0$-indexed ranking, that is, the best possible match has rank $0$. Motivated by~\citet{budish2009strategic}, we consider $ur(m_i, w_j) = n - r(m_i, w_j)$ (or, $ur(w_k, m_l) = n - r(w_k, m_l)$) to be the utility of man $m_i$ (or, woman $w_k$) on being matched to woman $w_j$ (or, man $m_l$). Let $\Sigma$ denote the set of all possible matchings.
%%%%%%%%%%%%%%%%%%%%%%%%%%%%%%
\subsection{Preliminaries}
\label{ssec:concepts}
%%%%%%%%%%%%%%%%%%%%%%%%%%%%%%
We denote match of an agent $a\in N$ under matching $\mathcal{M}$ as $\mathcal{M}_{a}$. In a matching $\mathcal{M}$, if $r(m_i,w_j)<r(m_i,\mathcal{M}_{m_i})$ and $r(w_j,m_i)$ $<r(w_j,\mathcal{M}_{w_j})$ , then $(m_i,w_j)$ would prefer to match with each other over the matching $\mathcal{M}$. Such a pair is called a \emph{blocking pair}~\citep{gale1962college}.
A matching $\mathcal{M}$ is said to be a \emph{stable matching} if it does not contain any \emph{blocking pairs}~\citep{gale1962college}. It is desirable to obtain a stable matching; otherwise, agents in blocking pairs may prefer to match with each other outside the market. It is a well-known result that for any arbitrary instance of $\succ$, a stable matching always exists~\citep{gale1962college}. Let $\Pi_{\succ}$ be the set of all possible stable matchings when the agents' original preference list is $\succ$.

A carefully constructed instance ($\succ$) may yield an exponential number of stable matchings (in $n$)~\citep{knuth1997stable}. The \emph{Gale-Shapley algorithm}\footnote{Gale and Shapley termed it as the ``deferred-acceptance'' algorithm. But in honor of their ground-breaking work, contemporary literature refers to it as the Gale-Shapley algorithm.}~\citep{gale1962college} outputs one such stable matching $\mathcal{M}^o(\succ)$ (assuming men propose to women). Similarly, it would output $\mathcal{M}^z(\succ)$ if women proposed to men. The stable matching $\mathcal{M}^o(\succ)$ is characteristic in that it is ``man-optimal'', i.e., each man gets the best possible match possible under the constraint of stability. In other words, $\forall m_i \in M, r(m_i, {\mathcal{M}^o}_{m_i}) = \min_{\mathcal{M}\in\Pi_{\succ}}r(m_i, {\mathcal{M}}_{m_i})$. The stable matching $\mathcal{M}^o$ is not only ``man-optimal'' but also ``women-pessimal'', i.e., $\forall w_j \in W, r(w_j, {\mathcal{M}^o}_{w_j}) = \max_{\mathcal{M}\in\Pi_{\succ}}r(w_j, {\mathcal{M}}_{w_j})$. Similarly, $\mathcal{M}^z$ is ``woman-optimal'' and ``men-pessimal''. Though $\mathcal{M}^{o}$ (or, $\mathcal{M}^{z}$) satisfies the criterion of stability, it is an unfair matching for $W$ (or, $M$). This concern has motivated research beyond the Gale-Shapley algorithm to introduce fairness in matching.

\noindent\textbf{Fairness in Matching.}
The set of all stable matchings forms a lattice~\citep{knuth1997stable}. It is desirable to determine a matching that is not only stable but also fair to all the agents. Towards this, there exist multiple optimal stable matchings in the literature - egalitarian stable matching~\citep{irving1987efficient}, minimum-regret stable matching~\citep{gusfield1987three}, and sex equal stable matching~\citep{kato1993complexity}. These stable matchings are optimal with respect to three fairness measures - egalitarian, regret, and sex equality.

\emph{Egalitarian criterion} is defined as the total sum of ranks of all matches.
\begin{equation*}
\label{eq:egalitarian-welfare}
    \mu^e(\mathcal{M}) = \sum_{m_i\in M}{r(m_i, \mathcal{M}_{m_i})}+\sum_{w_j\in W}{r(w_j, \mathcal{M}_{w_j})}
\end{equation*}

\emph{Regret} is defined as the worst individual dissatisfaction across men and women.
\begin{equation*}
    \label{eqn:regret}
    \mu^r(\mathcal{M}) = \max\Big(\max_{m_i \in M} {r(m_i, \mathcal{M}_{m_i})}, \max_{w_j\in W}{r(w_j, \mathcal{M}_{w_j})}\Big)
\end{equation*}

\emph{Sex-equality} is defined as the absolute difference between the total satisfaction of men and women.
\begin{equation*}
\label{eq:sex-equality}
    \mu^d(\mathcal{M}) = \Big\lvert\sum_{m_i\in M}{r(m_i, \mathcal{M}_{m_i})} - \sum_{w_j\in W}{r(w_j, \mathcal{M}_{w_j})}\Big\rvert
\end{equation*}
% \ra{Review with sir for definitions}
\begin{definition}[Egalitarian Stable Matching]~\citep{irving1987efficient}
\label{def:egalitarian-stable-matching}
A stable matching \( \mathcal{M}^e(\succ) \) is egalitarian if it minimizes the egalitarian criterion across all stable matchings.
    \begin{equation*}
    \label{eq:egalitarian-stable-matching}
    \mathcal{M}^e(\succ) =\arg\min_{\mathcal{M} \in \Pi_{\succ}}{\mu^e(\mathcal{M})}
    % \sum_{m_i\in M}{ur(m_i, \mathcal{M}_{m_i})}+\sum_{w_j\in W}{ur(w_j, \mathcal{M}_{w_j})}
    \end{equation*}    
\end{definition}

\begin{definition}[Minimum-Regret Stable Matching]~\citep{gusfield1987three}
\label{def:min-reg-stable-matching}
A stable matching \( \mathcal{M}^r(\succ) \) is minimum-regret if it minimizes the regret across all stable matchings.
    \begin{equation*}
    \label{eqn:minimum-regret-stable-matching}
    \mathcal{M}^r(\succ) = \arg\min_{\mathcal{M}\in \Pi_{\succ}}{\mu^r(\mathcal{M})}
    % \max\Big(\max_{m_i \in M} {mr(i, \mathcal{M}_{m_i})}, \max_{w_j\in W}{wr(j, \mathcal{M}_{w_j})}\Big)
    \end{equation*}
\end{definition}

\begin{definition}[Sex-Equal Stable Matching]~\citep{kato1993complexity}
\label{def:sex-equal-stable-matching}
A stable matching $\mathcal{M}^d(\succ)$ is sex equal if it minimizes the sex equality across all stable matchings.
    \begin{equation*}
    \label{eq:sex-equal-stable-matching}
        \mathcal{M}^d(\succ) = \arg\min_{\mathcal{M} \in \Pi_{\succ}}{\mu^d(\mathcal{M})}
        %{\lvert\sum_{m_i\in M}{mr(i, \mathcal{M}_{m_i})} - \sum_{w_j\in W}{wr(j, \mathcal{M}_{w_j})}\rvert}
    \end{equation*}
\end{definition}
\citet{irving1985efficient} utilize the concept of rotations and propose an $\mathcal{O}(n^4)$ algorithm that finds the egalitarian stable matching $\mathcal{M}^e(\succ)$. We refer to their algorithm as \emph{\theiralgo} (more details on this in Section~\ref{sssec:baseline-1}).~\citet{gusfield1987three} provides an algorithm to find $\mathcal{M}^r(\succ)$. However, finding $\mathcal{M}^d(\succ)$ is \texttt{NP-Hard}~\citep{kato1993complexity}.
Section~\ref{sec:notation} in the Appendix summarizes all the notation used in the paper.

\paragraph{Need for an improved approach. }
Consider the preference instance $\succ$ in Table~\ref{tab:succ-1}. The egalitarian stable matching $\mathcal{M}^e=\{(1, 4), (2, 1), (3, 3), (4, 5), (5, 2)\}$ performs poorly on $\mu^r = 4$. On the other hand, the minimum-regret stable matching $\mathcal{M}^r=\{(1, 4), (2, 2), (3, 1), (4, 5), (5, 3)\}$ performs poorly on $\mu^e = 11 = 3 + 8$. 

\begin{table}[ht]
    \caption{Example $\mathcal{M}^e$ dominate $\mathcal{M}^r$ on $\mu^e$ and $\mathcal{M}^r$ dominate $\mathcal{M}^e$ on $\mu^r$}
    \begin{minipage}{0.49\columnwidth}
        \centering
        \begin{tabular}{|p{1.2em}|p{1em}p{1em}p{1em}p{1em}p{1em}|}
            \hline
            \textbf{M} & \multicolumn{5}{c|}{\textbf{W}} \\
            \hline
            1: & 4 & 3 & 2 & 1 & 5 \\
            \hline
            2: & 5 & 4 & 2 & 3 & 1\\
            \hline
            3: & 5 & 1 & 3 & 4 & 2\\
            \hline
            4: & 5 & 2 & 1 & 4 & 3\\
            \hline
            5: & 3 & 2 & 4 & 5 & 1\\
            \hline
        \end{tabular}
        %\caption{$\succ_m$}
        \label{tab:men_pref-1}
    \end{minipage}
    \begin{minipage}{0.49\columnwidth}
        \centering
        \begin{tabular}{|p{1.2em}|p{1em}p{1em}p{1em}p{1em}p{1em}|}
            \hline
            \textbf{W} & \multicolumn{5}{c|}{\textbf{M}} \\
            \hline
             1: & 2 & 1 & 3 & 4 & 5\\ 
            \hline
             2: & 5 & 1 & 2 & 3 & 4\\ 
            \hline
             3: & 3 & 1 & 5 & 4 & 2\\ 
            \hline
             4: & 5 & 4 & 1 & 2 & 3\\ 
            \hline
             5: & 4 & 1 & 3 & 2 & 5\\
            \hline
        \end{tabular}
        %\caption{$\succ_w$}
        \label{tab:women_pref-1}
    \end{minipage}
    \label{tab:succ-1}
\end{table}

These gaps point to various limitations that we address in this paper. Our empirical analysis also shows that such examples shown above are not rare. Statistically, each of the three fairness measures above performs poorly on the others; we need a better measure that performs better overall. Towards this, we make use of the well-known measure of Nash social welfare (NSW)~\citep{nash1950bargaining} to address these limitations.
%%%%%%%%%%%%%%%%%%%%%%%%%%%%%%%%
\noindent\textbf{Nash social welfare.}
\label{ssec:nsw}
%%%%%%%%%%%%%%%%%%%%%%%%%%%%%%%%
\emph{Nash social welfare} (NSW) may be seen as fair in multiple ways in many other domains~\citep{nash1950bargaining, kaneko1979nash, moulin2004fair, caragiannis2019unreasonable}. \citet{moulin2004fair} shows that NSW is the only fair measure that satisfies many reasonable fairness axioms under general technical conditions. While the Nash bargaining solution maximizes the product of gains over a disagreement point, the NSW maximizes the product of utilities, which is equivalent to Nash bargaining with a disagreement point equal to zero. Accordingly, our formulation is as follows:

Nash social welfare (NSW) is the geometric mean of the utilities of men and women.
\begin{equation}
\label{eqn:nsw}
    \mu^{nsw}(\mathcal{M})=\bigg(\prod_{m_i\in M}{ur(m_i, \mathcal{M}_{m_i})}\times \prod_{w_j\in W}{ur(w_j, \mathcal{M}_{w_j})}\bigg)^{\frac{1}{2n}}
\end{equation}In two-sided market literature. For instance,~\citet{jain2024maximizing} work with agents' utilities for the matching and maximize Nash social welfare. We refer to such matching as \emph{NSW} matching. 
\begin{definition}
\label{def:nsw-matching}
    A matching \( \mathcal{M}^{nsw}(\succ) \) is Nash welfare optimal if it maximizes the Nash social welfare of all agents.
    \begin{equation*}
    \label{eq:nsw-matching}
    \mathcal{M}^{nsw}(\succ) =\arg\max_{\mathcal{M} \in \Sigma}{\mu^{nsw}(\mathcal{M})}
    \end{equation*}
\end{definition}
The matching $\mathcal{M}^{nsw}(\succ)$ can be computed in $\mathcal{O}(n^3)$ time using maximum-weight matching algorithms in bipartite graphs. The algorithm indeed outputs a matching that could be considered fair according to the aforementioned measures ($\mu^e, \mu^r, \mu^d$), owing to the larger search space. Anyhow, the matching it computes need not be stable. Hence, we propose to maximize NSW amongst all possible stable matchings. 
\begin{table}[ht]
    \caption{$\succ_2$: Nash social welfare performs well on all measures}
    \label{table:pref-2}
    \begin{minipage}{0.49\columnwidth}
        % \caption{$\succ_m$}
        \centering
        \begin{tabular}{|p{1.2em}|p{1em}p{1em}p{1em}p{1em}p{1em}|}
            \hline
            \textbf{M} & \multicolumn{5}{c|}{\textbf{W}} \\
            \hline
            1: & 3 & 2 & 5 & 1 & 4 \\
            \hline
            2: & 1 & 5 & 2 & 4 & 3 \\
            \hline
            3: & 4 & 3 & 1 & 2 & 5 \\
            \hline
            4: & 2 & 3 & 4 & 1 & 5 \\
            \hline
            5: & 2 & 3 & 5 & 1 & 4 \\
            \hline
            \end{tabular}
            % \label{tab:men_pref-2}
    \end{minipage}
    \hfill
    \begin{minipage}{0.49\columnwidth}
        % \caption{$\succ_w$}
        \centering
        \begin{tabular}{|p{1.2em}|p{1em}p{1em}p{1em}p{1em}p{1em}|}
            \hline
            \textbf{W} & \multicolumn{5}{c|}{\textbf{M}} \\
            \hline
             1: & 3 & 1 & 5 & 2 & 4 \\ 
            \hline
             2: & 1 & 3 & 4 & 2 & 5 \\ 
            \hline
             3: & 4 & 1 & 2 & 5 & 3 \\ 
            \hline
             4: & 4 & 1 & 2 & 5 & 3 \\ 
            \hline
             5: & 3 & 4 & 1 & 2 & 5 \\
            \hline
        \end{tabular}
        % \label{tab:women_pref-2}
    \end{minipage}
\end{table}
Consider $\succ_2$ given in Table~\ref{table:pref-2}. The stable matching $\mathcal{M}^e(\succ) = \{(1, 2), (2, 1), (3,$ $ 4), (4, 3), (5, 5)\}$ has $\mu^e = 15$, but it performs poorly on $\mu^d = 7$. The stable matching $\mathcal{M}^d(\succ) = \{(1, 3), (2, 5), (3, 1), (4, 2), (5, 4)\}$ has $\mu^d = 2$, but its $\mu^e = 16$. The stable matching $\mathcal{M} = \{(1, 2), (2, 5), (3, 4), (4, 3), (5,1)\}$ performs well on both measures: $\mu^e = 15$ and $\mu^d = 3$. This particular stable matching maximizes the Nash social welfare across all stable matchings.
%%%%%%%%%%%%%%%%%%%%%%%%%%%%%%%%%%%%%%%%%%%%%%%%%%%%%%%%
\section{Our Approach}
\label{sec:our-approach}
%%%%%%%%%%%%%%%%%%%%%%%%%%%%%%%%%%%%%%%%%%%%%%%%%%%%%%%%
%%%%%%%%%%%%%%%%%%%%%%%%%%%%%%%%%%%%%
\subsection{\ouralgo}
\label{ssec:snsw-alg}
%%%%%%%%%%%%%%%%%%%%%%%%%%%%%%%%%%%%%
First, we define an NSW optimal stable matching.

\begin{definition}
\label{def:snsw-stable-mathcing}
    A stable matching \( \mathcal{M}^{snsw}(\succ) \) is \emph{Nash welfare optimal} if it maximizes the Nash social welfare of all agents.
    \begin{equation*}
    \label{eq:snsw-stable-matching}
    \mathcal{M}^{snsw}(\succ) =\arg\max_{\mathcal{M} \in \Pi_{\succ}}{\mu^{nsw}(\mathcal{M})}
    \end{equation*}
\end{definition}

Note the difference in definitions of $\mathcal{M}^{nsw}(\succ)$ and $\mathcal{M}^{snsw}(\succ)$; the former is over all possible matching $\Sigma$ and the latter is over $\Pi_{\succ}$. Finding a matching $\mathcal{M}^{nsw}(\succ)$ is the same as maximum weight matching in a bipartite graph~\citep{jain2024maximizing}. However, because it does not need to be stable, we cannot leverage existing max-flow or matching algorithms to design an algorithm - \ouralgo\ that finds $\mathcal{M}^{snsw}$. On the other hand, \theiralgo\ in~\citet{irving1987efficient} determines a stable matching that minimizes $\mu^{e}$, that is, outputs $\mathcal{M}^e(\succ)$. We build upon \theiralgo. To this end, we present a short description of \theiralgo\ in the following paragraph. %Next, we describe how we leverage \theiralgo\ to adapt \ouralgo\ in the next subsection~\ref{ssec:our-algo}.
%%%%%%%%%%%%%%%%%%%%%%%%%%%%
\paragraph{\theiralgo.}
\label{sssec:baseline-1}
%%%%%%%%%%%%%%%%%%%%%%%%%%%%
\theiralgo\ begins with a stable matching (typically obtained via the Gale-Shapley algorithm) and identifies \emph{rotations}, following the framework of \citep{irving1985efficient}. 
\begin{definition}    
    A \emph{rotation} $\rho$,
    \[
    \rho={(m_0,w_0),\ldots,(m_{r-1},w_{r-1})},
    \]
    is a cyclic sequence of matched pairs, such that $\forall i \in [0, r - 1], w_i \succ_{m_i} w_{i + 1}$ and $m_{i - 1} \succ_{w_i} m_i$, where $(i + 1), (i - 1)$ is taken (mod $r$). Rotation (say $\rho$) is said to be exposed with respect to a \emph{reduced preference list} (say $\mathscr{S}$) if $w_i = first_{\mathscr{S}}(m_i) = second_{\mathscr{S}}(m_{i - 1})$.
\end{definition}

Successive elimination of rotations produces ``reduced preference lists'' $\mathscr{S}$. A rotation $\rho$ is said to be \emph{eliminated} if, $\forall i \in [0, r - 1],\forall j:m_j \succ_{w_i} m_{i - 1}, m_j\not\in\mathscr{S}_{w_i}, w_i \not\in \mathscr{S}_{m_j}$. If a rotation $\rho$ cannot be exposed before another rotation $\pi$ is eliminated, then $\pi$ is said to \emph{precede} $\rho$. The resulting precedence structure defines the \emph{rotation poset} $P$, whose vertices ($V(P)$) correspond to rotations and whose directed edges ($E(P)$) represent precedence relations~\citep{irving1985efficient}. We define $pred_P({\rho})$ as all predecessors of rotation $\rho$ in $P$.

A sparse subgraph $P'$ of $P$ is constructed using the following rules:
\begin{enumerate}
\item [\bf R1] If $(m_{i^*},w_{j^*})$ belongs to rotation $\pi$, and $w_{j'}$ is the first woman below $w_{j^*}$ in $m_{i^*}$'s preference list such that $(m_{i^*},w_{j'})$ belongs to rotation $\rho$, then add edge $(\pi,\rho)$.
\item [\bf R2] If $(m_{i^*},w_{j'})$ does not belong to any rotation but is eliminated by rotation $\pi$, and $w_{j^*}$ is the first woman above $w_{j'}$ such that $(m_{i^*},w_{j^*})$ belongs to rotation $\rho$, then add edge $(\pi,\rho)$.
\end{enumerate}

The transitive closure of $P'$ equals $P$; consequently, $P'$ preserves all closed subsets and therefore all stable matchings~\citep{irving1986complexity}. Specifically, eliminating the rotations of a closed subset yields a reduced preference structure $\mathscr{S}$, and the stable matching $(m, first_{\mathscr{S}}(m))$ corresponds uniquely to that subset.

Given $P'$, a capacitated $s$-$t$ network $P'(s,t)$ is constructed. A source $s$ is connected to each negative-weight rotation $\rho$ with capacity $|w(\rho)|$, while each positive-weight rotation is connected to sink $t$ with capacity $w(\rho)$. All original edges of $P'$ are assigned infinite capacity. The following result links minimum cuts and maximum-weight closed subsets.

\begin{theorem}[\citep{irving1987efficient}]
\label{thm:1}
Let $X$ be a minimum $s$-$t$ cut in $P'(s,t)$. The positive rotations whose edges to $t$ remain uncut, together with all their predecessors in $P'$, form a maximum-weight closed subset of $P'$.
\end{theorem}

Since the number of rotations and edges in $P'$ are each bounded by $\mathcal{O}(n^2)$, both the network size and maximum flow value are $\mathcal{O}(n^2)$. Therefore, applying Ford-Fulkerson yields the maximum-weight closed subset—and hence the egalitarian-optimal stable matching—in $\mathcal{O}(|K||E|)=\mathcal{O}(n^4)$ time.
%%%%%%%%%%%%%%%%%%%%%%%%%%%%
\paragraph{Adapting \theiralgo\ to build \ouralgo.}
\label{sssec:our-algo}
%%%%%%%%%%%%%%%%%%%%%%%%%%%%
Note that \theiralgo\ assigns a weight $w(\rho)$ (Eqn.~\ref{eqn:wrho}) to a rotation $\rho=\{(m_0, w_0), (m_1, w_1), \ldots, (m_{r-1}, w_{r-1})\}$ in poset $P$. With this, the algorithm actually generates a new graph $P'$.
The goal in the algorithm is to minimize $\mu^e(M)$. 

\begin{equation}
\label{eqn:wrho}
    w(\rho) = \sum_{i=0}^{r-1}{{ur(m_i, \rho_{m_{i+1}})}-{ur(m_i, \rho_{m_i})}} +\sum_{i=0}^{r-1}{{ur(w_i,\rho_{w_{i-1}})}-{ur(w_i,\rho_{w_i})}}
\end{equation}

Our goal is to maximize Nash social welfare, which is the product of utilities. Using logarithmic transformations, we can reduce the problem to maximizing the sum of the logarithms of utilities. Hence, we propose to use weights of rotations as defined in Definition~\ref{def:wt-rot}.

\begin{definition}
\label{def:wt-rot}
    Given a rotation $\rho=\{(m_0, w_0), (m_1, w_1), \ldots, (m_{r-1}, w_{r-1})\}$ in a stable marriage instance, we define the weight $w'(\rho)$ of that rotation by
    \begin{equation}
    \label{eq:wt-rot}
    w'(\rho) = \sum_{i=0}^{r-1}{\log{ur(m_i, \rho_{m_{i+1}})}-{\log{ur(m_i, \rho_{m_i})}}} + \sum_{i=0}^{r-1}{\log{ur(w_i,\rho_{w_{i-1}})}-\log{ur(w_i,\rho_{w_i})}}
    \end{equation}
\end{definition}

\noindent Determining $\mathcal{M}^{snsw}(\succ)$ reduces to determining $\mathcal{M}^{e}(\succ)$ by plugging the weights in Equation~\ref{eq:wt-rot} in \theiralgo.

In \theiralgo, the edge weights of $P'(s, t)$ were integers. Hence, they could apply the Ford-Fulkerson algorithm to find the max-flow in $\mathcal{O}(C\cdot E)$ time. In \ouralgo, edge weights need not be integers. Therefore, we employ Dinic's maximum flow algorithm~\citep{dinitz2006dinitz} with dynamic trees improvement by~\citet{sleator1981data} to find the max-flow in \(\mathcal{O}(|V|\cdot|E|\cdot\log|V|)\) time.

Thus, \theiralgo\ adapted to our settings is described as follows.
\begin{algorithm}[!ht]
\caption{\ouralgo}
\label{alg:ouralgo}
\begin{algorithmic}[1]
\State \textbf{Input:} $\succ = ((\succ_a)_{a\in N})$
\State \textbf{Output:} $\mathcal{M}^{snsw}(\succ)$
\State \textbf{Procedure:}
\State $\mathcal{M}^o(\succ_m, \succ_w)$ = \texttt{Gale-Shapley}($\succ$)
\State Determine the vertex set $V(P)$
\State Initialize $w(\rho),\: \forall \rho\in V(P)$ using Equation~\ref{eq:wt-rot}
\State Construct subgraph $P'$ via rules in Section~\ref{ssec:snsw-alg}
\State $\forall e \in E(P'), w(e) \leftarrow \infty$, Construct $P'(s, t)$ as follows:
\State $\forall \rho\in V(P'): (i) w(\rho)<0, e:s\rightarrow\rho, w(e):= |w(\rho)|, (ii) w(\rho) > 0, e:\rho\rightarrow t, w(e):= w(\rho)$ 
\State $X = min-cut(P'(s, t))$ (Dinic's algorithm~\citep{dinitz2006dinitz} with dynamic trees improvement by~\citet{sleator1981data})
\State $X^+:=\{\rho\in P'(s, t):e(\rho, t) \text{ uncut by } X\}, \mathfrak{S}:= X^+\cup pred_P(X^+)$
\ForAll{$\rho \in \mathfrak{S}$}
    \State Eliminate $\rho$, and Update $\mathscr{S}$ as explained in Section~\ref{sssec:baseline-1}
\EndFor
\State return $\mathcal{M}^{\ouralgo}(\succ)=\{(m_i, first_{\mathscr{S}}(m_i))\}_{m_i\in M}$
\end{algorithmic}
\end{algorithm}

\noindent Now we prove the correctness and the time complexity of \ouralgo.

%%%%%%%%%%%%%%%%%%%%%%%%%%
\subsection{Theoretical Analysis of \ouralgo}
\label{ssec:thry_ana}
%%%%%%%%%%%%%%%%%%%%%%%%%%
We prove the correctness of \ouralgo\ as follows.
\begin{proposition}
    Let $\mathcal{M}'$ be the stable matching obtained on eliminating rotations $\rho_1, \rho_2, \ldots, \rho_t$ from the man-oriented shortlist, then
    \begin{equation*}
    \label{eq:cor-proof}
        \mu^{nsw}(\mathcal{M}') = \mu^{nsw}(\mathcal{M}^o) * \left(\prod_{i=1}^{t}{\exp{\frac{w'(\rho_i)}{2n}}}\right)
    \end{equation*}
\end{proposition}
\begin{proof}
\label{proof:4}
Consider a stable matching $\mathcal{M}^s$, with rotation $\rho$ exposed in its preference lists, and let $\mathcal{M}$ be the stable matching obtained upon eliminating $\rho$ from the reduced preference lists associated with $\mathcal{M}^s$. Thus, the following follows naturally.
    \begin{align*}
    \label{eqns:claim-4}
        &w'(\rho) = \sum_{i=0}^{r-1}{\log{ur(m_i, \rho_{m_{i+1}})}-{\log{ur(m_i, \rho_{m_i})}}} + \sum_{i=0}^{r-1}{\log{ur(w_i,\rho_{w_{i-1}})}-\log{ur(w_i,\rho_{w_i})}} \qquad \text{Using}~\ref{eq:wt-rot}\\
        &\exp{w'(\rho)} = \prod_{i=0}^{r-1}{\frac{ur(m_i, \rho_{m_{i+1}})\cdot ur(w_i,\rho_{w_{i-1}})}{ur(m_i, \rho_{m_{i}})\cdot ur(w_i,\rho_{w_{i}})}}\\
        &\text{Multiply numerator and denominator by the utility of all man-woman pairs who remain matched after elimination of $\rho$}\\
        &\exp{w'(\rho)} =\frac{\bigg(\prod_{m_i\in M}{ur(m_i, \mathcal{M}_{m_i})}\times \prod_{w_j\in W}{ur(w_j, \mathcal{M}_{w_j})}\bigg)}{\bigg(\prod_{m_i\in M}{ur(m_i, \mathcal{M}^s_{m_i})}\times \prod_{w_j\in W}{ur(w_j, \mathcal{M}^s_{w_j})}\bigg)} = \left(\frac{\mu^{nsw}(\mathcal{M})}{\mu^{nsw}(\mathcal{M}^s)}\right)^{2n}
    \end{align*}
    Since $\mathcal{M}'$ is obtained upon the successive elimination of $\rho_1, \ldots, \rho_t$ from shortlists associated with $\mathcal{M}^o$, the equation associated with each rotation may be multiplied together to prove the proposition.
\end{proof}
\begin{theorem}
\label{thm:ouralgo-snsw}
    For every preference profile $\succ$, \ouralgo\ outputs the stable matching that maximizes the Nash social welfare, i.e., \ouralgo\ returns $\mathcal{M}^{snsw}(\succ)$.
\end{theorem}
\begin{proof}
\label{proof:correctness}
    \ouralgo\ outputs the maximum-weight closed subset $\mathfrak{S}$ in the weighted rotation poset $P$ using $w(\rho)$ from Def.~\ref{def:wt-rot}. 
    \begin{equation*}
    \label{eq:nsw-rot-elim}
        \mu^{nsw}(\mathcal{M}^{\ouralgo}) = \mu^{nsw}(\mathcal{M}^o) * \left(\prod_{\rho_i\in\mathfrak{S}}{\exp{\frac{w'(\rho_i)}{2n}}}\right)
    \end{equation*}
     With this, we need that $\mathcal{M}^{snsw}(\succ)$ is obtained on eliminating the rotations in $\mathfrak{S}$ starting from the man-oriented shortlists to complete the proof.
    It follows from the Theorem~\ref{thm:1} and Definition~\ref{def:snsw-stable-mathcing} that $\mathcal{M}^{\ouralgo}(\succ)=\mathcal{M}^{snsw}(\succ)$. 
\end{proof}
We will now present the proof of the running time complexity of \ouralgo.
\begin{theorem}
\label{thm:our-algo-time}
The running time complexity of \ouralgo\ is $\mathcal{O}(n^4\log n)$.
\end{theorem}
\begin{proof}
\label{proof:time}
Note that   $P'$ is a subgraph of the rotation poset $P$. Thus, $V(P')\subseteq V(P)$. Since there exist at most $\mathcal{O}(n^2)$ man-woman pairs, and each pair is part of at most one rotation, we have  $|V(P')| \le |V(P)| = \mathcal{O}(n^2)$.
   \begin{observation} 
     $|V| $ in $P'(s, t)$ is bounded by $\mathcal{O}(n^2)$.
     \label{obs:one}
    \end{observation}

        Next, we observe that the rules R1 and R2 used to create a subgraph $P'$ of $P$ associate each edge with a man-woman pair. Furthermore, each pair is associated with the creation of an edge at most once. Therefore, $|E(P')| \le n*n = \mathcal{O}(n^2)$. We formalize this in the next observation. 

   \begin{observation}
       $|E|$ in $P'(s, t)$ is bounded by $\mathcal{O}(n^2)$.
       \label{obs:two}
   \end{observation}

Finally, we observe that the graph $P'(s, t)$ as adopted in \ouralgo\ has positive edge weights, and the time complexity of its construction is $\mathcal{O}$. This observation directly follows from the construction of $P'(s, t)$ from $P'$. All edges from the source $s$ to the negative weight rotation nodes ($\rho$) have capacities $|w(\rho)|$.
    All edges from positive weight rotation nodes ($\pi$) to the sink node $t$ have capacities $w(\pi)$.
    All other edges in $P'(s, t)$ borrowed from $P'$ have capacities set to infinity.
    Therefore, all edges have positive edge weights. 
    Moreover, since the number of vertices in $P'$ is upper-bounded by $\mathcal{O}(n^2)$ from Observation 1.
    
    \begin{observation}
    The graph $P'(s, t)$ as adopted in \ouralgo\ has positive edge weights, and its construction is $\mathcal{O}(n^2)$ time complexity.
    \label{obs:three}
    \end{observation}

    \noindent {\em Everything put together:} Finding the max-flow using Dinic's algorithm~\citep{dinitz2006dinitz} with dynamic trees improvement by~\citet{sleator1981data} requires $\mathcal{O}(|V|\cdot|E|\cdot\log|V|)$ time. Based on Observations 1, 2, and 3, this time complexity reduces to $\mathcal{O}(n^4\log n)$, or $\tilde{\mathcal{O}}(n^4)$. 
\end{proof}
%%%%%%%%%%%%%%%%%%%%%%%%%%%%%%%%%%%%%%%%%%%%%%%%%%%%%%%%
\section{Empirical Evaluation}
\label{sec:expr}
%%%%%%%%%%%%%%%%%%%%%%%%%%%%%%%%%%%%%%%%%%%%%%%%%%%%%%%%
We empirically evaluate the fairness of our proposed algorithm, \ouralgo, on randomly generated data. The goal of our experiments is to demonstrate the fairness of our approach relative to existing baselines across multiple dimensions. We begin by describing the experimental setup, followed by how we generate our data, the baseline algorithms we compare \ouralgo\ against, and the fairness metrics used to evaluate the matching. Finally, we describe the experiments performed and the observations.
%%%%%%%%%%%%%%%%%%%
\subsection{Setup}
\label{ssec:setup}
%%%%%%%%%%%%%%%%%%%
%%%%%%%%%%%%%%%%%%%%%%%%%%%%%%%%%%
\paragraph{Preference Generation.}
\label{sssec:data}
%%%%%%%%%%%%%%%%%%%%%%%%%%%%%%%%%%
We study four distributions over the space of preference profiles $\succ$, each capturing different structural assumptions about the input. We may categorize them into two groups: 
\begin{itemize}
    \item[i] Uniform Distribution $\mathrm{D}_1$: $\succ_m \sim \mathcal{U}[n]$ and $\succ_w \sim \mathcal{U}[n]$, that is, preference is any permutation drawn uniformly at random.
    \item[ii] Popularity-based Distributions~\citep{zou2010tolerable}. Here, we initialize a popularity profile $p_a$ of all $a\in N$ according to \begin{itemize}
        \item[a] $ \mathrm{P}_{\mathcal{U}}$, where $p_a \sim\mathcal{U}[0, 1]$
        \item[b] $\mathrm{P}_{\mathcal{T}}$, where $p_a \sim \mathcal{T}[0, 1]$ is triangular with peak at $0.5$.
        \item[c] $\mathrm{P}_{\mathcal{N}}$, where $p_a \sim \mathcal{N}(\mu=0, \sigma=1)$ is half-normal.
    \end{itemize}
    Next, we sample the first preference of $m_i$ ($w_j$) from $W$ ($M$) as $\mathbb{P}_{w_x} = \frac{p(w_x)}{\sum_{w_z\in W}{p(w_z)}}$ \bigg($\mathbb{P}(m_k)=\frac{p(m_k)}{\sum_{m_l\in M}{p(m_l)}}$\bigg). Once the sampled woman (man) appends to $\succ_m$ ($\succ_w$), we remove it from $W$ ($M$) and resample the next preference. This is performed iteratively till $\succ_a$ is completed for agent $a$.
\end{itemize}
\begin{figure}
    \begin{minipage}[t]{0.949\textwidth}
    \centering
        \captionof{figure}{Circular Plots for $n=25$}
        \label{fig:2}
        \begin{subfigure}{0.49\textwidth}
            \centering
            \includegraphics[width=\linewidth]{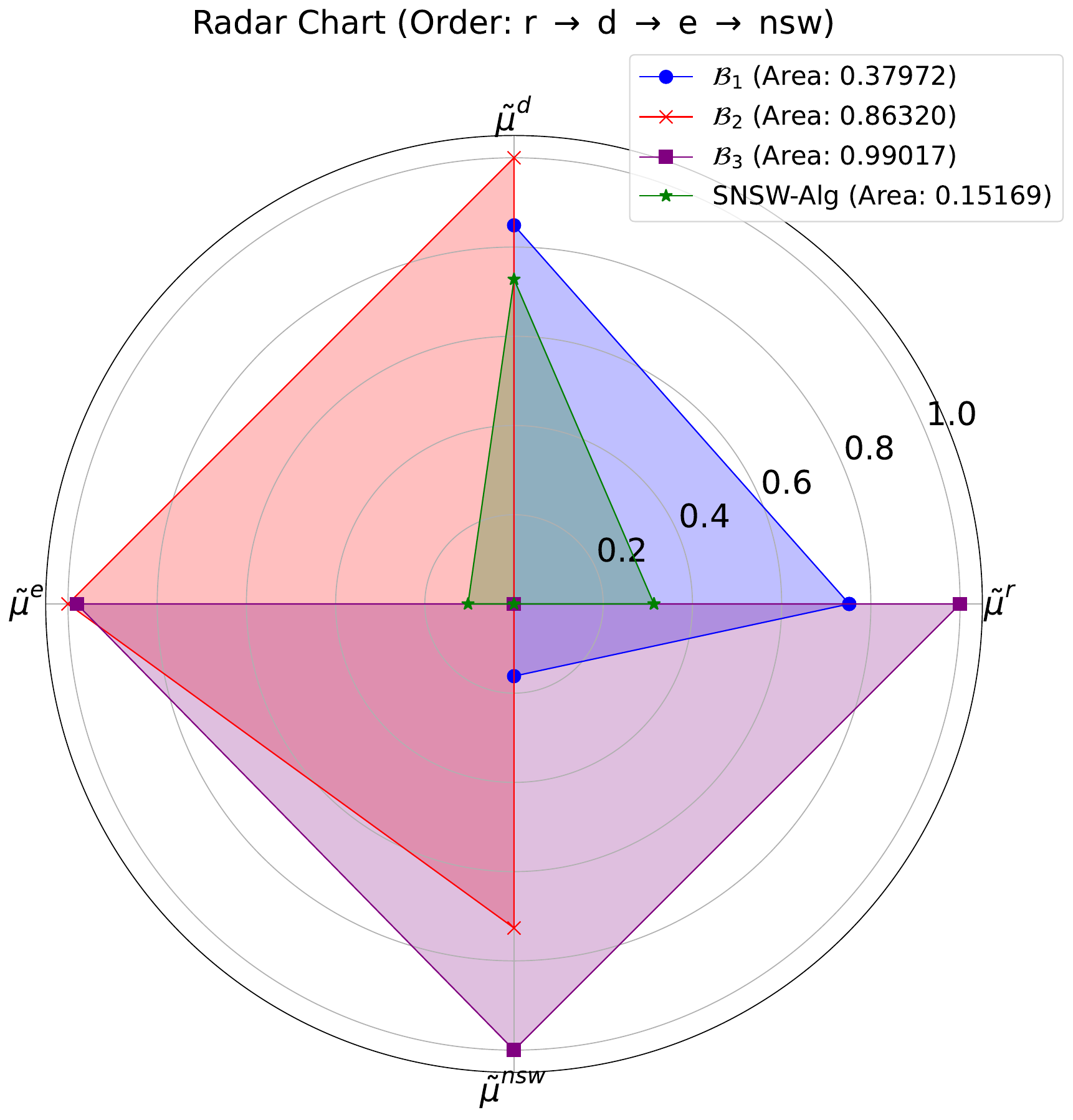}
            \caption{$\mathcal{U}$}
            \label{fig:exp-main-paper-6}
        \end{subfigure}
        \hfill
        \begin{subfigure}{0.49\textwidth}
            \centering
            \includegraphics[width=\linewidth]{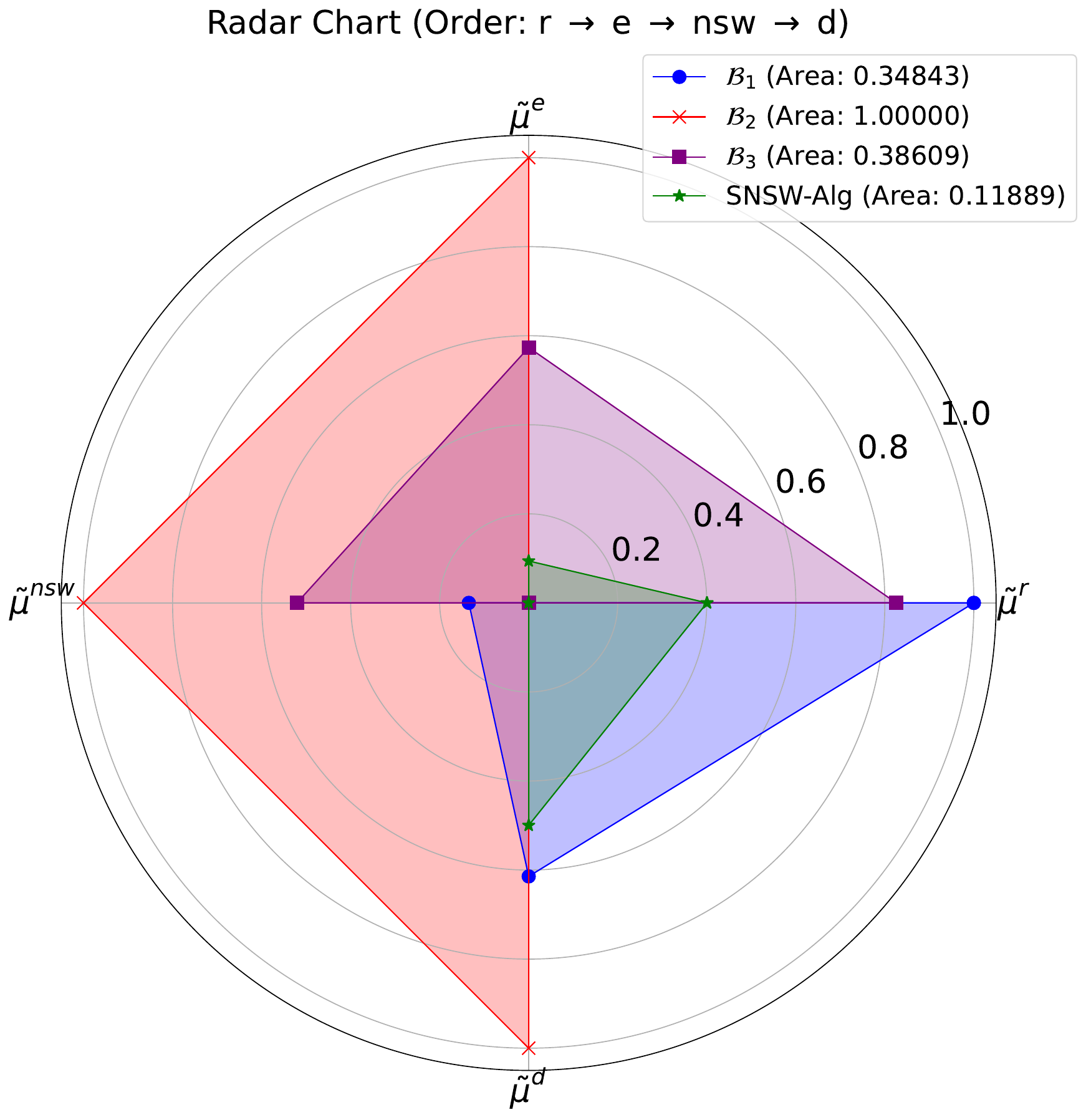}
            \caption{$\mathrm{P}_{\mathcal{U}}$}
            \label{fig:exp-main-paper-7}
        \end{subfigure}
        
        \begin{subfigure}{0.49\textwidth}
            \centering
            \includegraphics[width=\linewidth]{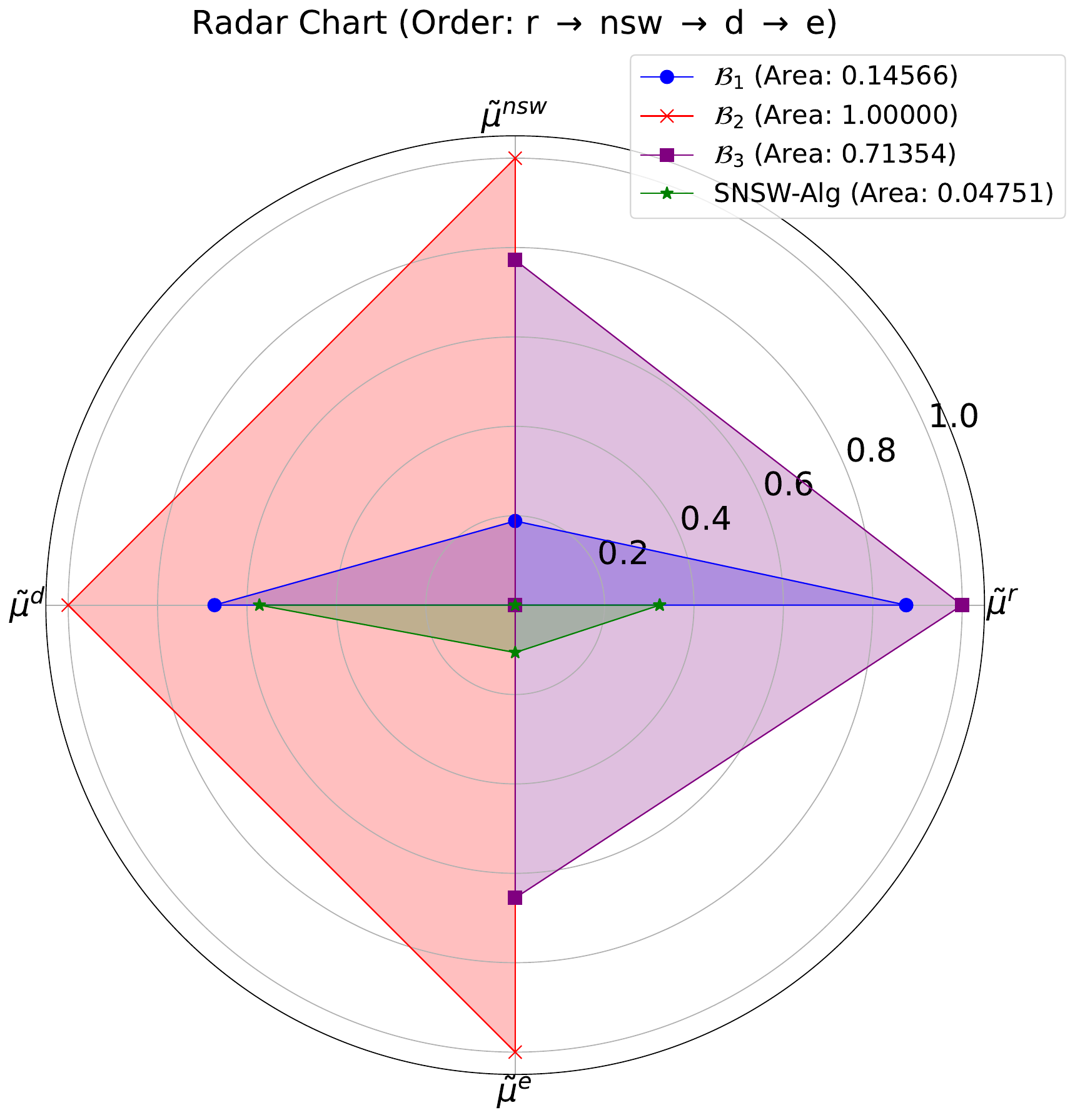}
            \caption{$\mathrm{P}_{\mathcal{T}}$}
            \label{fig:exp-main-paper-8}
        \end{subfigure}
        \hfill
        \begin{subfigure}{0.49\textwidth}
            \centering
            \includegraphics[width=\linewidth]{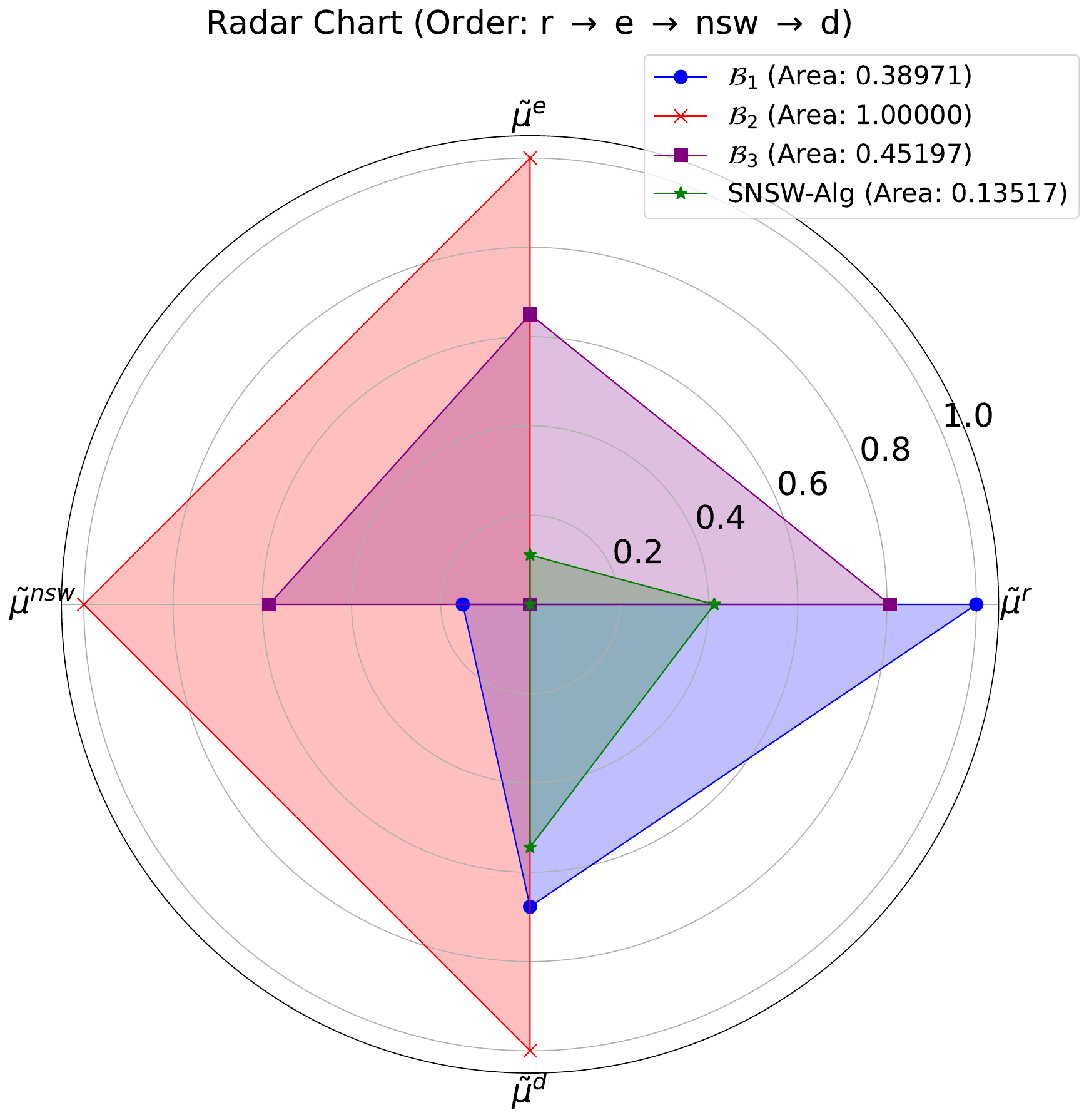}
            \caption{$\mathrm{P}_{\mathcal{N}}$}
            \label{fig:exp-main-paper-9}
        \end{subfigure}
    \end{minipage}
\end{figure}
%%%%%%%%%%%%%%%%%%%%%%%%%%
\paragraph{Benchmarks.}
\label{sssec:alg}
%%%%%%%%%%%%%%%%%%%%%%%%%%
We consider the following baselines to compare with \ouralgo.
\begin{itemize}
    \item [i]$\mathcal{B}_1$: \theiralgo~\citep{irving1987efficient}
    \item [ii]$\mathcal{B}_2$: \texttt{Min-Regret-Alg}~\citep{gusfield1987three}
    \item [iii]$\mathcal{B}_3$: \texttt{Sex-equal-Alg} iterates over all possible stable matchings and determines the stable matching with optimal sex equality measure since the sex equal stable matching problem is \texttt{NP-Hard}~\citep{kato1993complexity}. The distributions we considered, on average, did not yield an exponential number of stable matchings ($|\Pi|$). Thus, we were able to complete all experiments within $6.5$ hours, using a 13th Gen Intel i7-1355U (2.3 GHz) with 16GB of RAM.
\end{itemize}
%%%%%%%%%%%%%%%%%%%%%%%%%%%%%%%%%%
\paragraph{Evaluation Metrics.}
\label{sssec:measure}
%%%%%%%%%%%%%%%%%%%%%%%%%%%%%%%%%%
We observe that while Baselines $\mathcal{B}_1$, $\mathcal{B}_2$, and $\mathcal{B}_3$ achieve strong performance on their respective objectives, their performance deteriorates considerably on the other fairness metrics. By comparison, \ouralgo\ achieves a more balanced performance profile across all four metrics. There exist instances of $\succ$ such that \ouralgo\ outputs $\mathcal{M}^{snsw}$ that performs relatively well across all four measures $\mu^e, \mu^r, \mu^d, \mu^{nsw}$. For example, for $\succ_2$ given in Table~\ref{table:pref-2}, $\mathcal{M}^{snsw}(\succ)$ performs particularly well. While such examples are of theoretical importance, it is a worst-case example. Hence, instead of evaluating algorithms on worst-case performance, we employ statistical versions~\ref{eqns:statistical-measures}. We compute the empirical mean measures of all baselines across 100K instances.

\begin{equation}
\label{eqns:statistical-measures}
    \begin{split}
    \tilde{\mu}^{e, ALG} &=\frac{\sum_{i=1}^{K}{\mu^e(\mathcal{M}^{ALG}_i)}}{K} \quad
    \tilde{\mu}^{r, ALG} =\frac{\sum_{i=1}^{K}{\mu^r(\mathcal{M}^{ALG}_i)}}{K} \\
    \tilde{\mu}^{d, ALG} &=\frac{\sum_{i=1}^{K}{\mu^d(\mathcal{M}^{ALG}_i)}}{K} 
    \quad \tilde{\mu}^{nsw, ALG} =\frac{\sum_{i=1}^{K}{\mu^{nsw}(\mathcal{M}^{ALG}_i)}}{K} \\
    \end{split}
\end{equation}
\paragraph{Setup}
\label{sssec:setup}
%%%%%%%%%%%%%%%%%%%%%
We use $K = 100,000$ instances for $n=25$ using all four distributions $\mathcal{U}, \mathrm{P}_{\mathcal{U}}, \mathrm{P}_{\mathcal{T}}, \mathrm{P}_{\mathcal{N}}$ and compare $\tilde{\mu}^{e, ALG}, \tilde{\mu}^{r, ALG}, \tilde{\mu}^{d, ALG}, \tilde{\mu}^{nsw, ALG}$. To compare these statistical versions uniformly, we normalize them to $[0,1]$, $0$ being the best.

Consider the stable matching produced by algorithm $ALG$. Let the statistical regret measured be $\tilde{\mu}^{r, ALG}$. Let the maximum and minimum values of $\tilde{\mu}^r$ measured across all baselines be $r_{max}, r_{min}$. Then we normalize $\tilde{\mu}^{r, ALG}$ to $\frac{\tilde{\mu}^{r, ALG} - r_{min}}{r_{max}-r_{min}}$. We normalize $\tilde{\mu}^{nsw, ALG}$ to $\frac{r_{max}-\tilde{\mu}^{nsw, ALG}}{r_{max}-r_{min}}$. Each baseline would achieve $0$ on the measure it optimizes. The smaller the triangle's area, the fairer the algorithm's matching is.

\begin{table}[t]
\centering
\caption{CoV across different Measures}
\label{tab:std_results}
\begin{tabular}{llcccc}
\toprule
Agents & Metric & Regret & Egalitarian & Disparity & SNSW \\
\midrule
\multirow{4}{*}{25}
& $\mu^r$   & 0.060 & 0.055 & 0.055 & 0.058 \\
& $\mu^e$   & 0.094 & 0.113 & 0.095 & 0.113 \\
& $\mu^d$   & 0.518 & 0.521 & 0.597 & 0.523 \\
& $\mu^{nsw}$ & 0.062 & 0.061 & 0.061 & 0.062 \\
\midrule
\multirow{4}{*}{50}
& $\mu^r$ & 0.035 & 0.032 & 0.032 & 0.034 \\
& $\mu^e$ & 0.082 & 0.093 & 0.082 & 0.092 \\
& $\mu^d$ & 0.587 & 0.597 & 0.696 & 0.600 \\
& $\mu^{nsw}$ & 0.045 & 0.045 & 0.044 & 0.045 \\
\bottomrule
\end{tabular}
\end{table}

%%%%%%%%%%%%%%%%%%%%%%%%%%%%%%%%%
\subsection{Observations}
\label{ssec:res}
%%%%%%%%%%%%%%%%%%%%%%%%%%%%%%%%%

We first consider $\ouralgo$ against baselines for all 6 possible combinations of two fairness measures at a time. We observe in Figure~\ref{fig:exp-main-paper-5} that \ouralgo\ is Pareto-undominated in such pairwise measure comparison for $n = 25$. We measure the value of Vargha-Delaney $A^{12}$ measure comparing $\ouralgo$'s output with the baseline algorithms on all measures other than those it optimizes over. We observe that the measure is strictly greater than $0.5$ for all such combinations.

We also compare all four measures simultaneously.
We plot $\tilde{\mu}$'s of all the algorithms on a circle (Figure~\ref{fig:2}). We observe that \ouralgo\ has the least area: at least 93.5\% smaller than all benchmarks, demonstrating the statistical significance of \ouralgo\ in achieving fairness across multiple criteria.
\begin{figure}
    \centering
    \caption{Pareto Undominance of \ouralgo\ against other baselines}
    \includegraphics[width=0.85\linewidth]{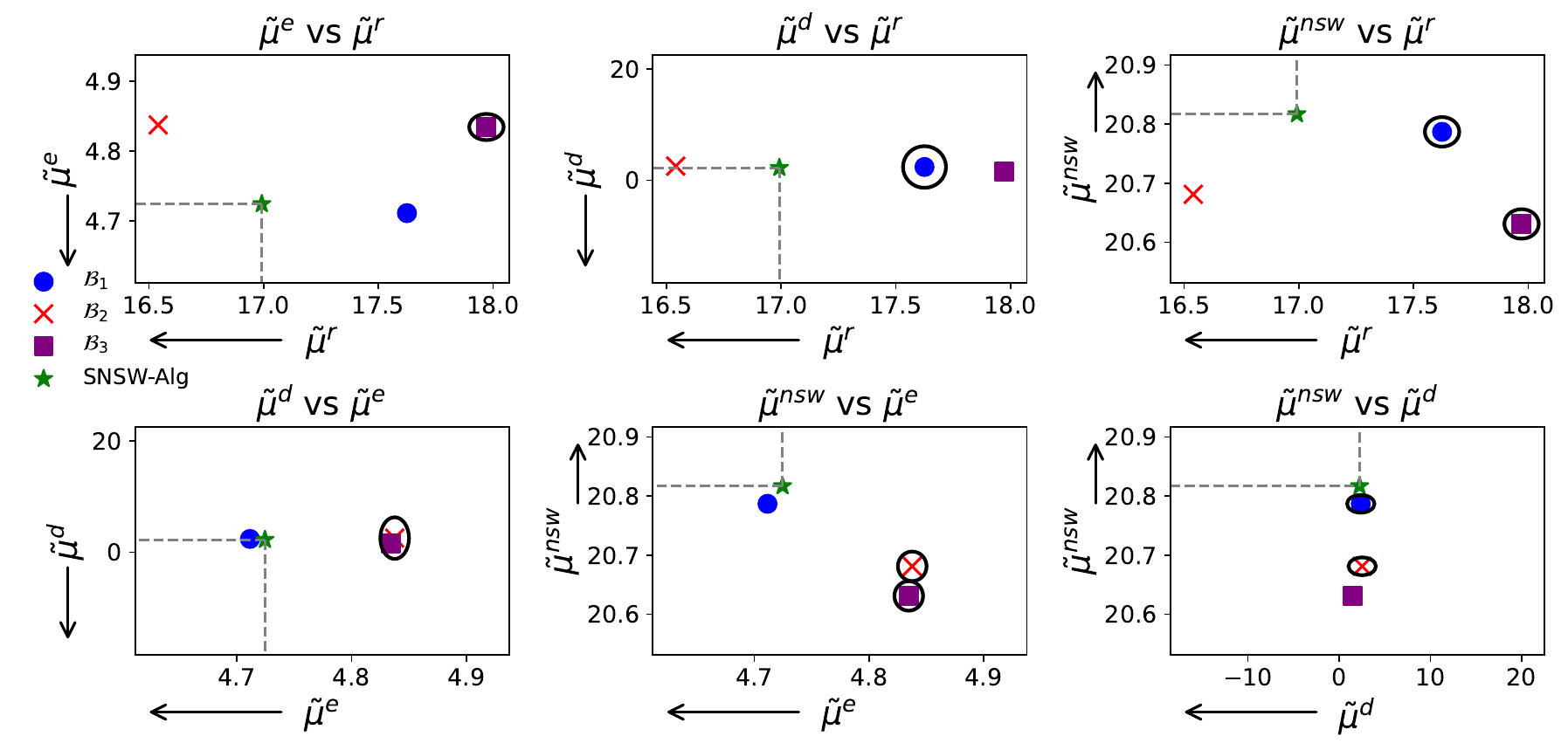}
    \label{fig:exp-main-paper-5}
\end{figure}
%%%%%%%%%%%%%%%%%%%%%%%%%%%%%%%%%%%%%%%%%%%%
\section{Related Work}
\label{sec:related}
%%%%%%%%%%%%%%%%%%%%%%%%%%%%%%%%%%%%%%%%%%%%
%%%%%%%%%%%%%%%%%%%%%%%%%%%%%%%%%%
\textbf{Matching Theory.}
\label{ssec:rel-matchings}
%%%%%%%%%%%%%%%%%%%%%%%%%%%%%%%%%%
~\citet{gale1962college} proposed an algorithm for assigning students to colleges (or residents to hospitals) along with preserving the stability of the matching.~\citet{roth2004kidney} proposed a design to build upon kidney exchange practices, taking into account properties of efficiency and incentive compatibility.~\citet{delacretaz2016refugee} proposed several refugee resettlement mechanisms that improve match efficiency as well as incentivize refugees to report their choices, and respect the priorities of local areas.
More recently,~\citet{aziz2022matching} authored a survey on many-to-one matching under various market constraints, such as lower quotas, regional constraints, diversity constraints, multidimensional capacity constraints, and hereditary constraints.~\citet{kamiyama2020popular} studied the problem of popular matching under matroidal constraints, and proposed a polynomial-time algorithm to find the largest size popular matching under the given constraints.

%%%%%%%%%%%%%%%%%%%%%%%%%%%%%%%%%%
\noindent\textbf{Fairness in Stable Marriage Problem.}
\label{ssec:rel-fairness}
%%%%%%%%%%%%%%%%%%%%%%%%%%%%%%%%%%
\citet{irving1987efficient} established a fairness criterion to measure how equitable any given stable matching is in the context of the stable marriage problem, and proposed an $\mathcal{O}(n^4)$ algorithm to find such an optimal (egalitarian) stable matching.~\citet{gusfield1987three} proposed another fairness metric for the stable marriage problem called regret and proposed an $\mathcal{O}(n^2)$ algorithm to find the minimum-regret stable matching.~\citet{gusfield1989parametric} studied the parametric stable marriage problem and analyzed the nature of the parametrized measure and time complexity of finding the optimal solution to the parametric stable marriage problem.~\citet{kato1993complexity} studied the complexity of another problem - sex equal stable marriage, which aimed to minimize the net disparity between the assigned matches to the two sides, concluding that this problem is \texttt{NP-Hard}. 
More recently,~\citet{alkan2022equitable} introduced ``modular stable matching rules'' as a framework for equity in stable matchings, establishing an equivalence between modular optimization and an ordinal condition called convexity. They also propose a new equity criterion called ``equity undominance'' and characterize the modular rules that satisfy it.

%%%%%%%%%%%%%%%%%%%%%%%%%%%%%%%%%%
\noindent\textbf{Nash social welfare.}
\label{ssec:rel-nsw}
%%%%%%%%%%%%%%%%%%%%%%%%%%%%%%%%%%
\citet{branzei2017nash} study the question of (approximately) implementing the objective of Nash social welfare for fair allocation of resources in the presence of strategic agents in relation to one-sided markets.~\citet{barman2019fair} propose algorithms for fair allocation of indivisible goods to strategic agents in single-parameter environments along with theoretical Nash welfare guarantees.~\citet{barman2022nash} utilizes the Nash social welfare function to study fairness and efficiency in coverage problems.~\citet{jain2024maximizing} studied the complexity of problems in many-to-one matching (not necessarily stable) that aim to maximize Nash social welfare in two-sided markets.~\citet{garg2025approximating} study the interplay between competitive equilibrium and Nash welfare maximizing allocations in the context of the divisible items allocation problem. To summarize, Nash social welfare has been studied extensively in one-sided markets. As per our knowledge, the work of~\citet{jain2024maximizing} is the sole work on the application of Nash social welfare in two-sided markets.
%%%%%%%%%%%%%%%%%%%%%%%%%%%%%%%%%%%%%%%%%%%%%%%%%%%%
\section{Conclusion}
\label{sec:conclusion}
%%%%%%%%%%%%%%%%%%%%%%%%%%%%%%%%%%%%%%%%%%%%%%%%%%%%
In this paper, we studied the problem of obtaining a fair and stable matching in two-sided markets. Matching produced to optimize any of the existing measures performs poorly on other measures. To this end, we proposed optimizing Nash social welfare under rank-induced utilities while ensuring stability. We then designed \ouralgo\ to output the NSW optimal stable matching. We then proved the correctness of \ouralgo\ and showed its running time complexity is $\tilde{\mathcal{O}}(n^4)$. We then empirically study the fairness of matching with different approaches. Empirically, we demonstrate on 100K randomly generated instances per setting, across uniform and popularity-based preference distributions, that, overall, the NSW-optimal stable matching statistically outperforms egalitarian, regret, or sex-equality. Our study shows that the proposed measure $\mu^{nsw}$, coupled with stability, yields a fair outcome. We leave further theoretical investigations into its guarantees relative to other measures as future work.
\bibliography{references}
\bibliographystyle{plainnat}
\appendix
%%%%%%%%%%%%%%%%%%%%
\section{Notation}
\label{sec:notation}
%%%%%%%%%%%%%%%%%%%%
\begin{table}[H]
\label{tab:notation}
\centering
\begin{tabular}{|l|l|}
\hline
    $M$ & Men side of market: \(\{m_1, m_2, \ldots, m_n\}\) \\
\hline
    $W$ & Women side of market: \(\{w_1, w_2, \ldots, w_n\}\) \\
\hline
    $n$ & Number of men and women \(\mid M \mid = \mid W \mid = n\)\\
\hline
    $N$ & Set of all men and women: $M\cup W$ \\
\hline
    $r(a, b)$ & Rank of agent $b$ in $\succ_{a}$ \\
\hline
    $ur(m_i, w_j)$ & Utility of man $m_i$ on match with woman $w_j$\\
\hline
    $ur(w_k, m_l)$ & Utility of woman $w_k$ on match with man $m_l$\\
\hline
    $\mathcal{M}$ & Any stable matching between $M$ and $W$\\
\hline
    $\mathcal{M}^o$ & Man-optimal stable matching \\
\hline
    $\mathcal{M}^z$ & Woman-optimal stable matching \\
\hline
    $\mathcal{M}^r$ & Min-Regret stable matching \\
\hline
    $\mathcal{M}^e$ & Egalitarian stable matching \\
\hline
    $\mathcal{M}^d$ & Sex-equal stable matching \\
\hline
    $\mathcal{M}^{nsw}$ & Nash optimal matching \\
\hline
    $\mathcal{M}^{snsw}$ & Nash stable matching \\
\hline
\end{tabular}
\end{table}

\section{Experimental Evaluations}

%%%%%%%%%% Uniform Plots %%%%%%%%
\begin{figure}[ht]
    \centering

    \begin{subfigure}{0.48\textwidth}
        \centering
        \includegraphics[width=\textwidth]{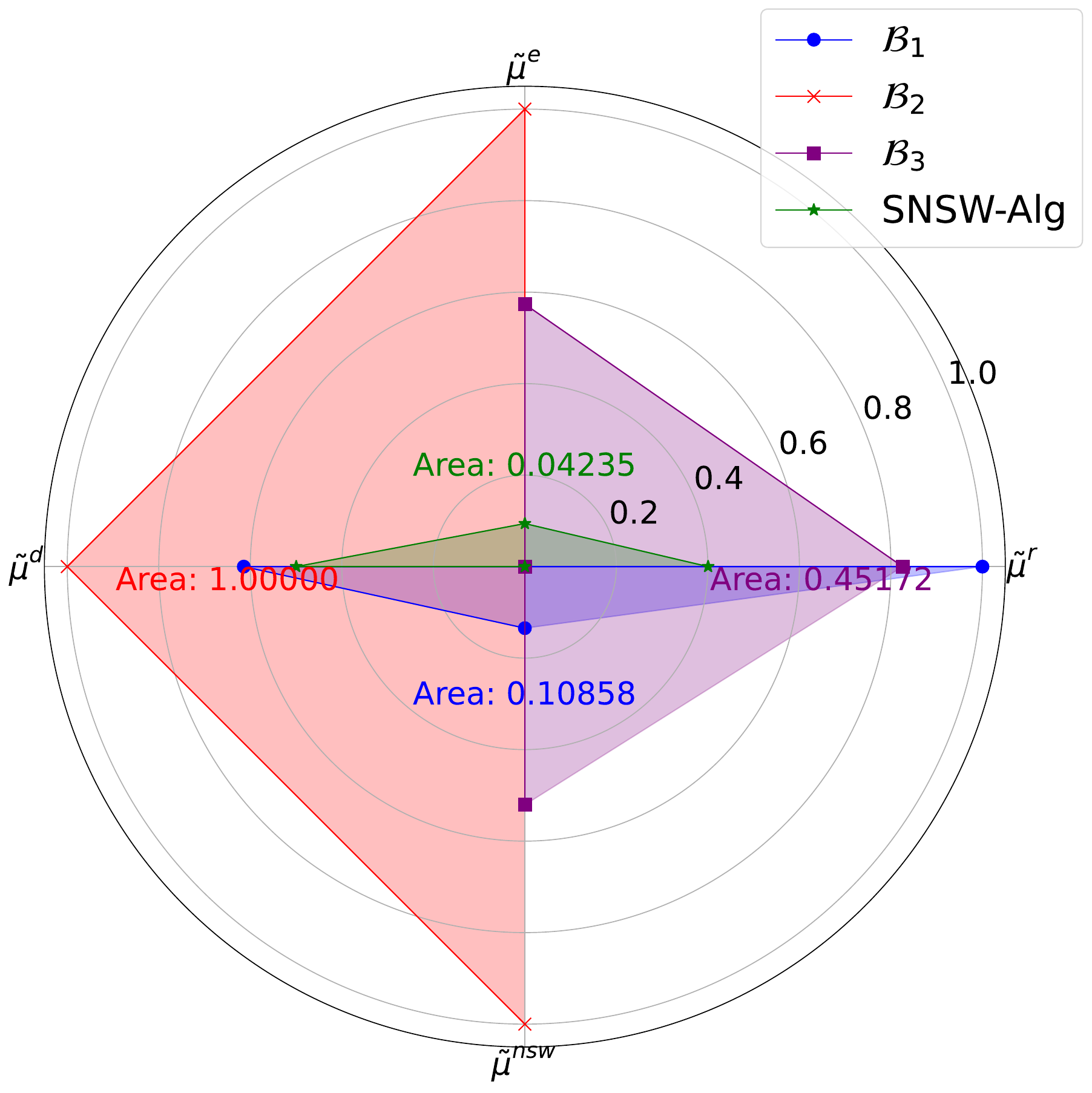}
        % \caption{First figure}
        \label{fig:fig1_popularity_uniform}
    \end{subfigure}
    \hfill
    \begin{subfigure}{0.48\textwidth}
        \centering
        \includegraphics[width=\textwidth]{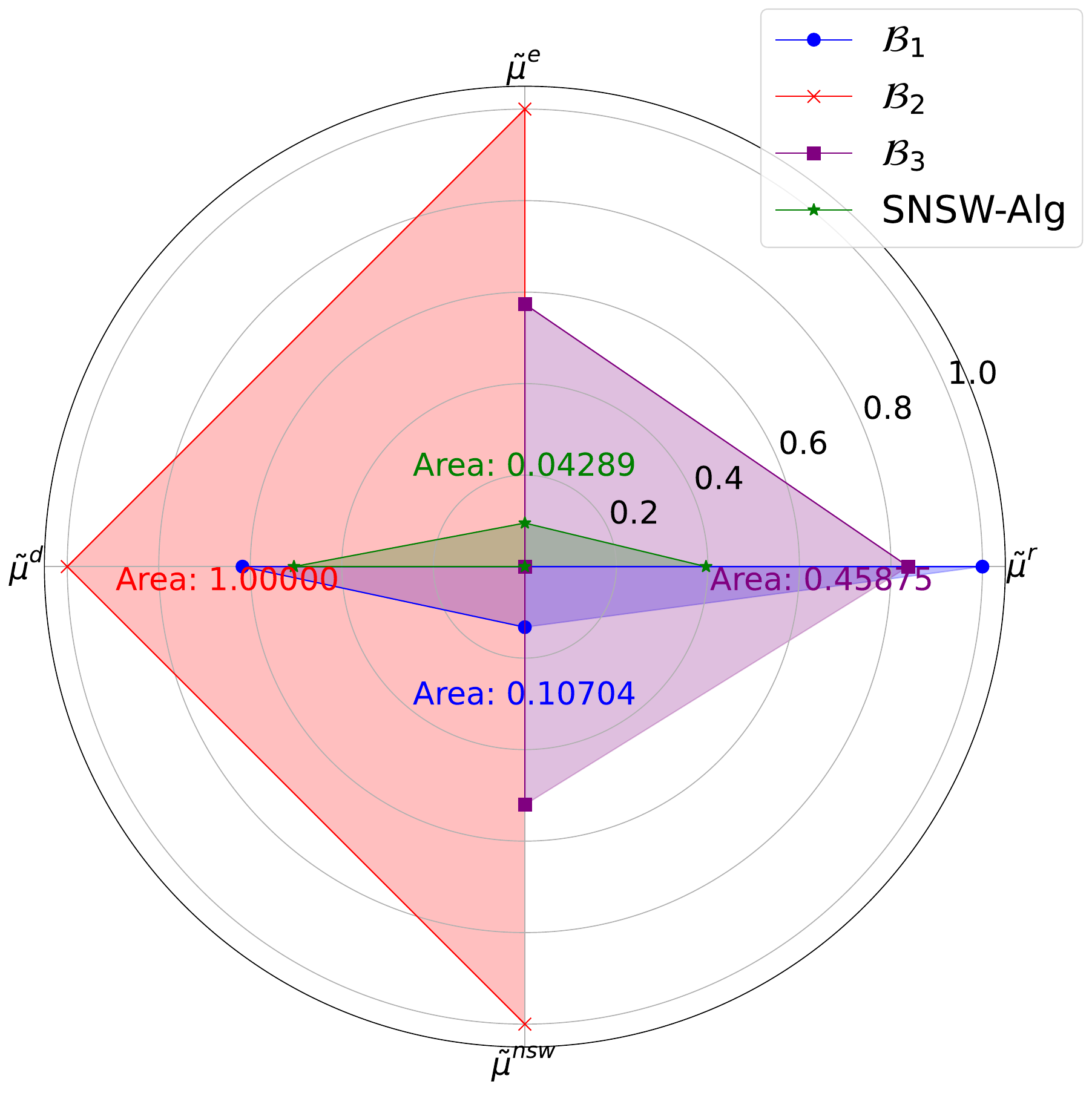}
        % \caption{Second figure}
        \label{fig:fig2_popularity_uniform}
    \end{subfigure}

    \begin{subfigure}{0.48\textwidth}
        \centering
        \includegraphics[width=\textwidth]{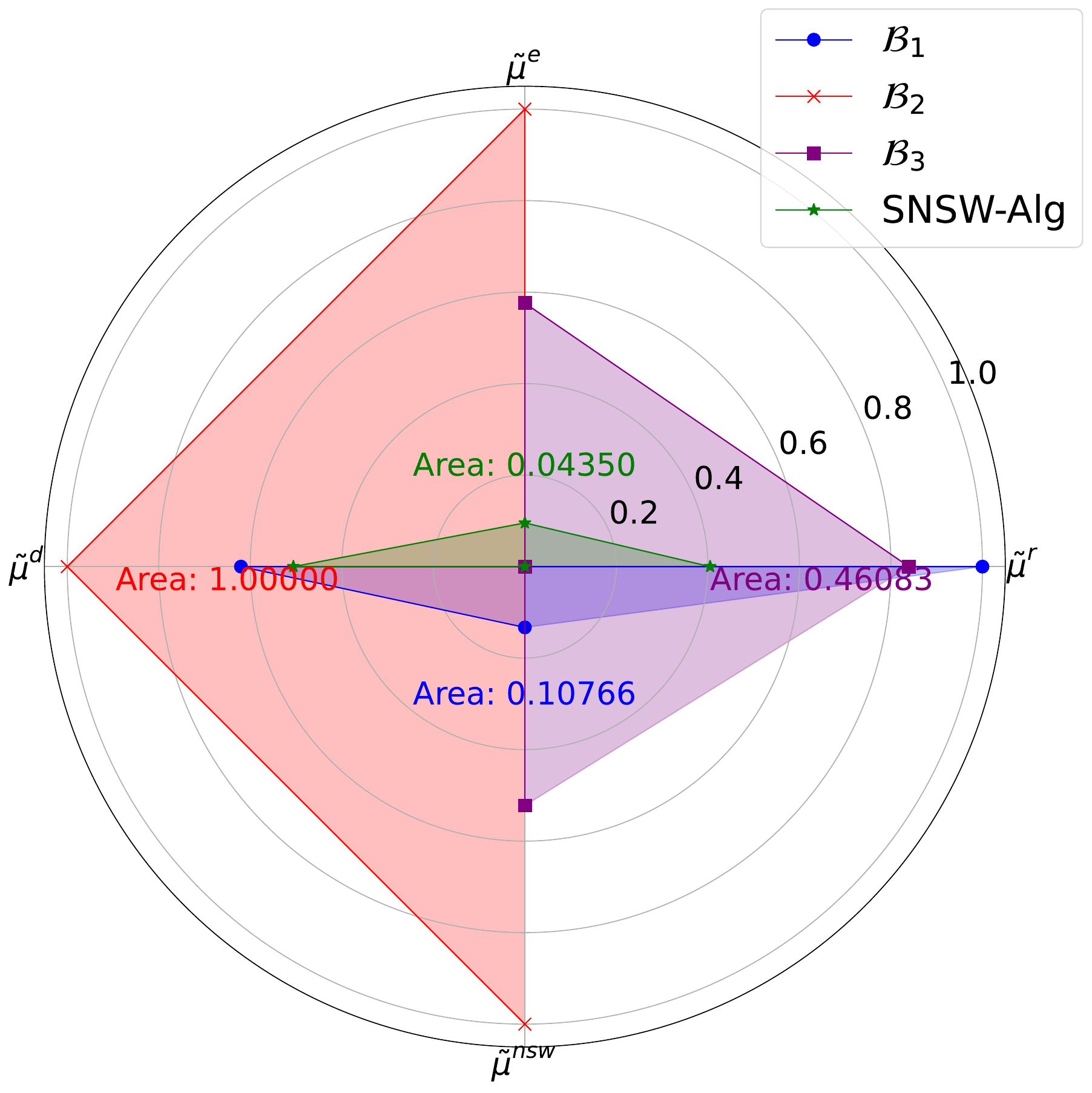}
        % \caption{Third figure}
        \label{fig:fig3_popularity_uniform}
    \end{subfigure}
    \hfill
    \begin{subfigure}{0.48\textwidth}
        \centering
        \includegraphics[width=\textwidth]{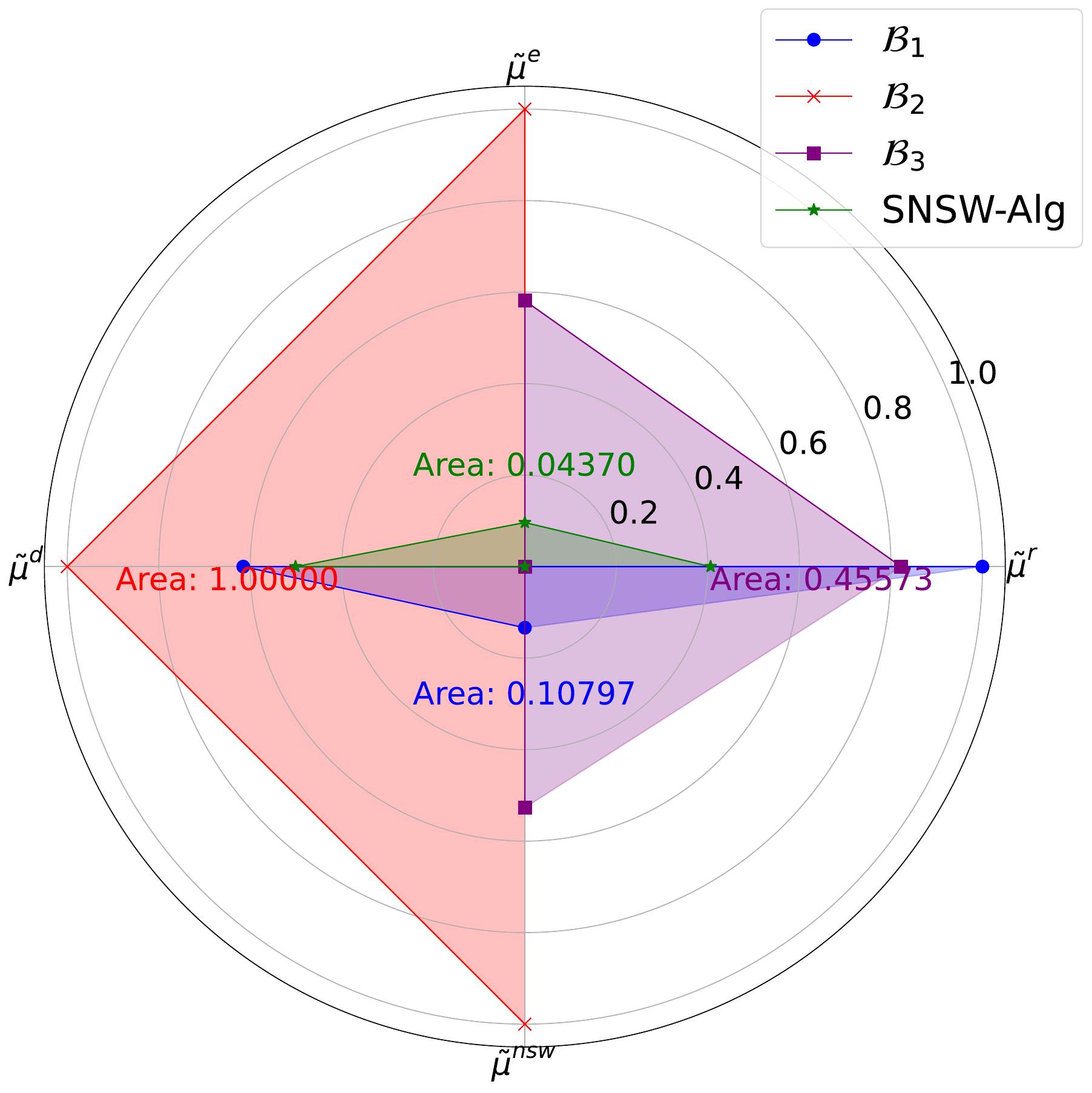}
        % \caption{Fourth figure}
        \label{fig:fig4_popularity_uniform}
    \end{subfigure}

    \caption{Circular Plot for Uniform Distribution}
    \label{fig:four_figures_popularity_uniform}
\end{figure}

%%%%%%%%%% Triangular Plots %%%%%%%%
\begin{figure}[ht]
    \centering

    \begin{subfigure}{0.48\textwidth}
        \centering
        \includegraphics[width=\textwidth]{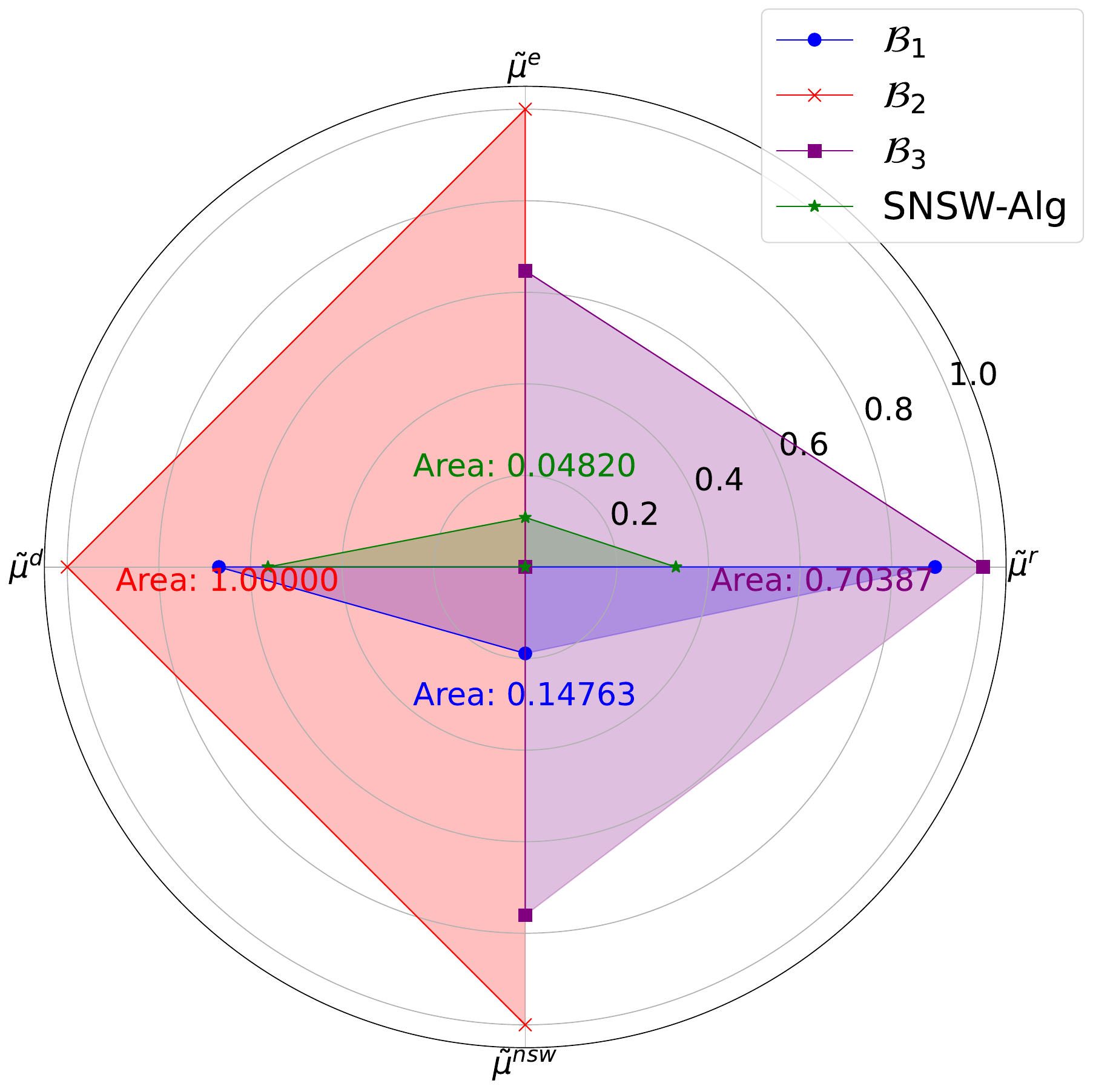}
        % \caption{First figure}
        \label{fig:fig1_popularity_triangular}
    \end{subfigure}
    \hfill
    \begin{subfigure}{0.48\textwidth}
        \centering
        \includegraphics[width=\textwidth]{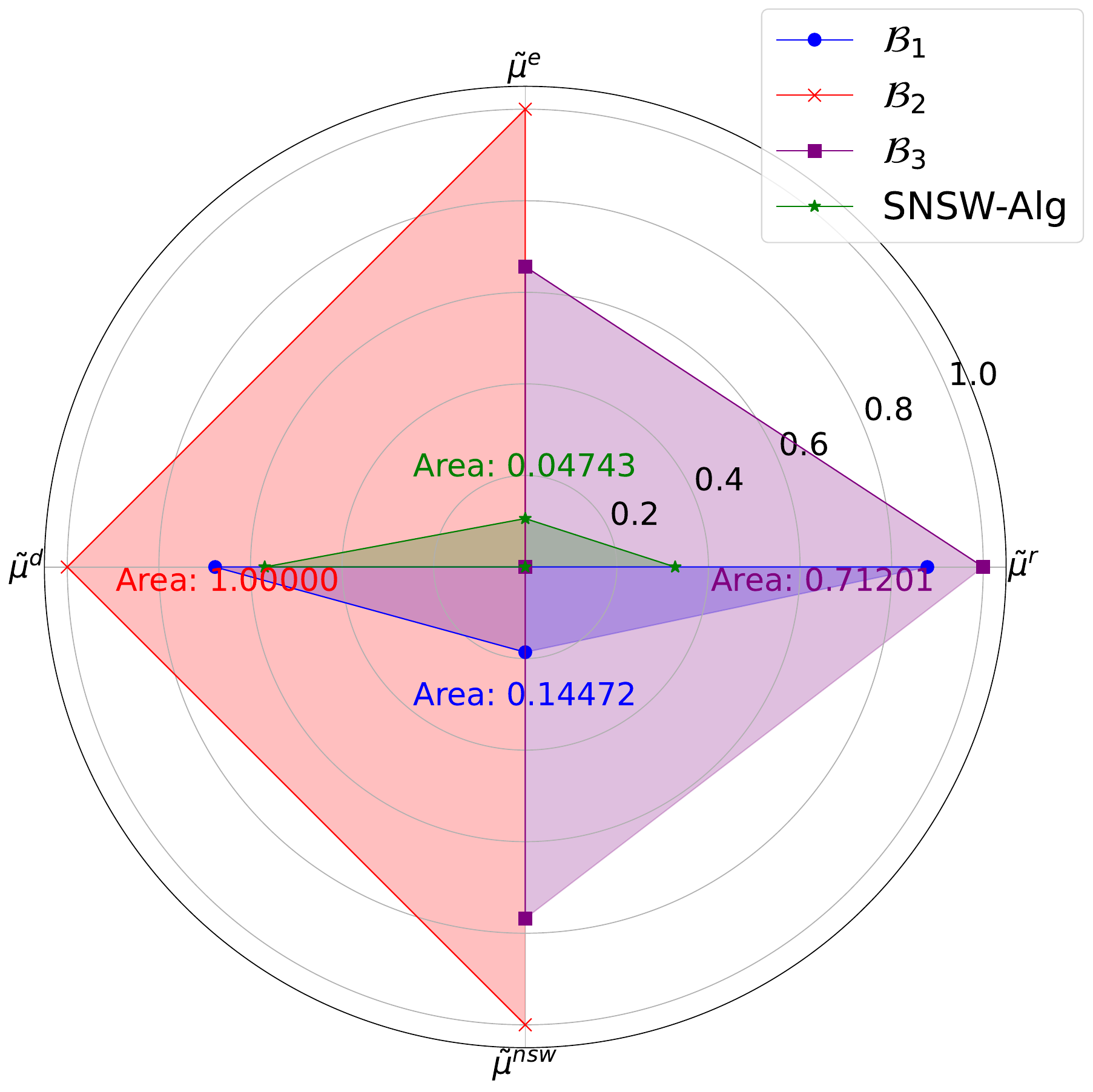}
        % \caption{Second figure}
        \label{fig:fig2_popularity_triangular}
    \end{subfigure}

    \begin{subfigure}{0.48\textwidth}
        \centering
        \includegraphics[width=\textwidth]{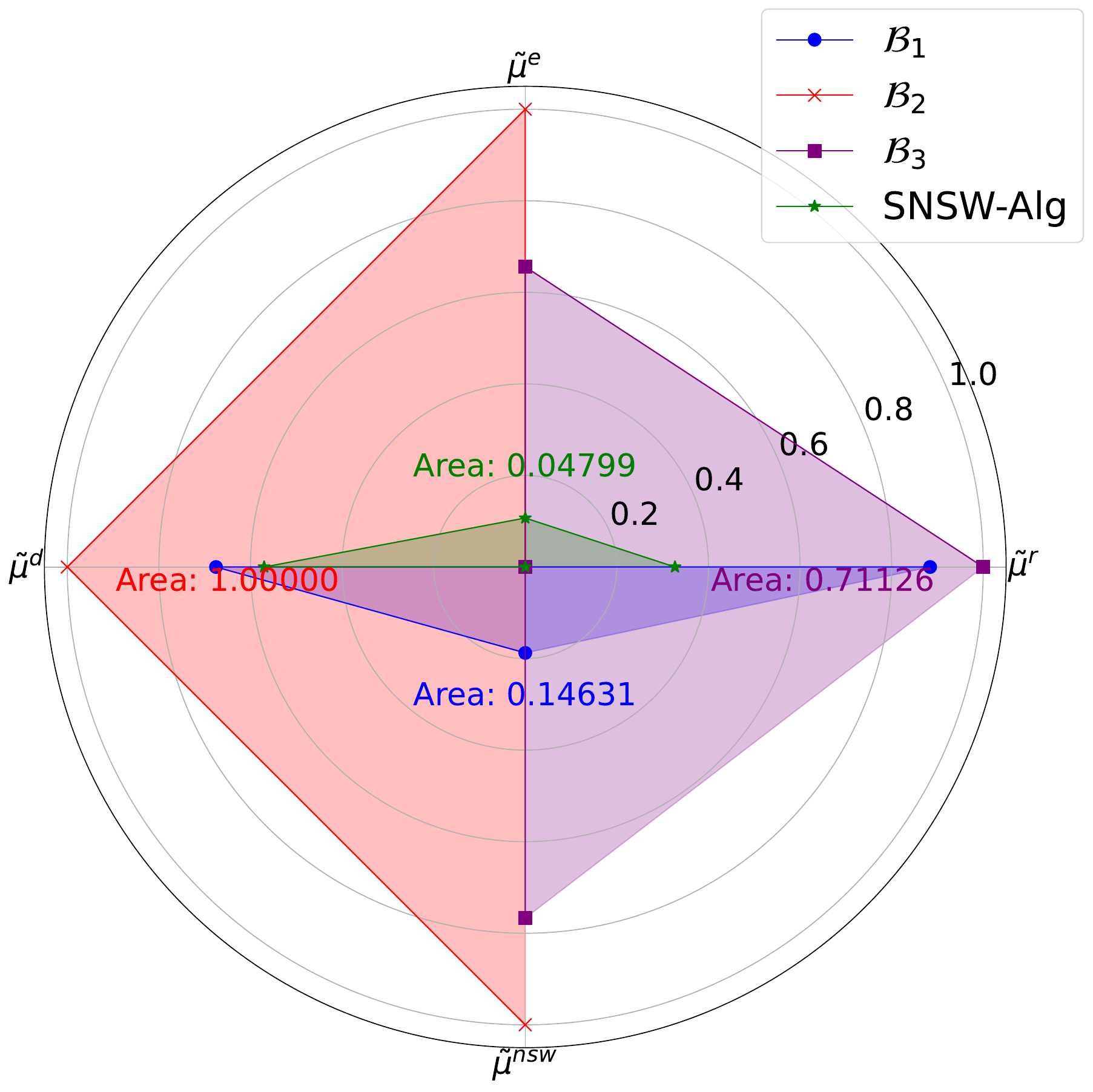}
        % \caption{Third figure}
        \label{fig:fig3_popularity_triangular}
    \end{subfigure}
    \hfill
    \begin{subfigure}{0.48\textwidth}
        \centering
        \includegraphics[width=\textwidth]{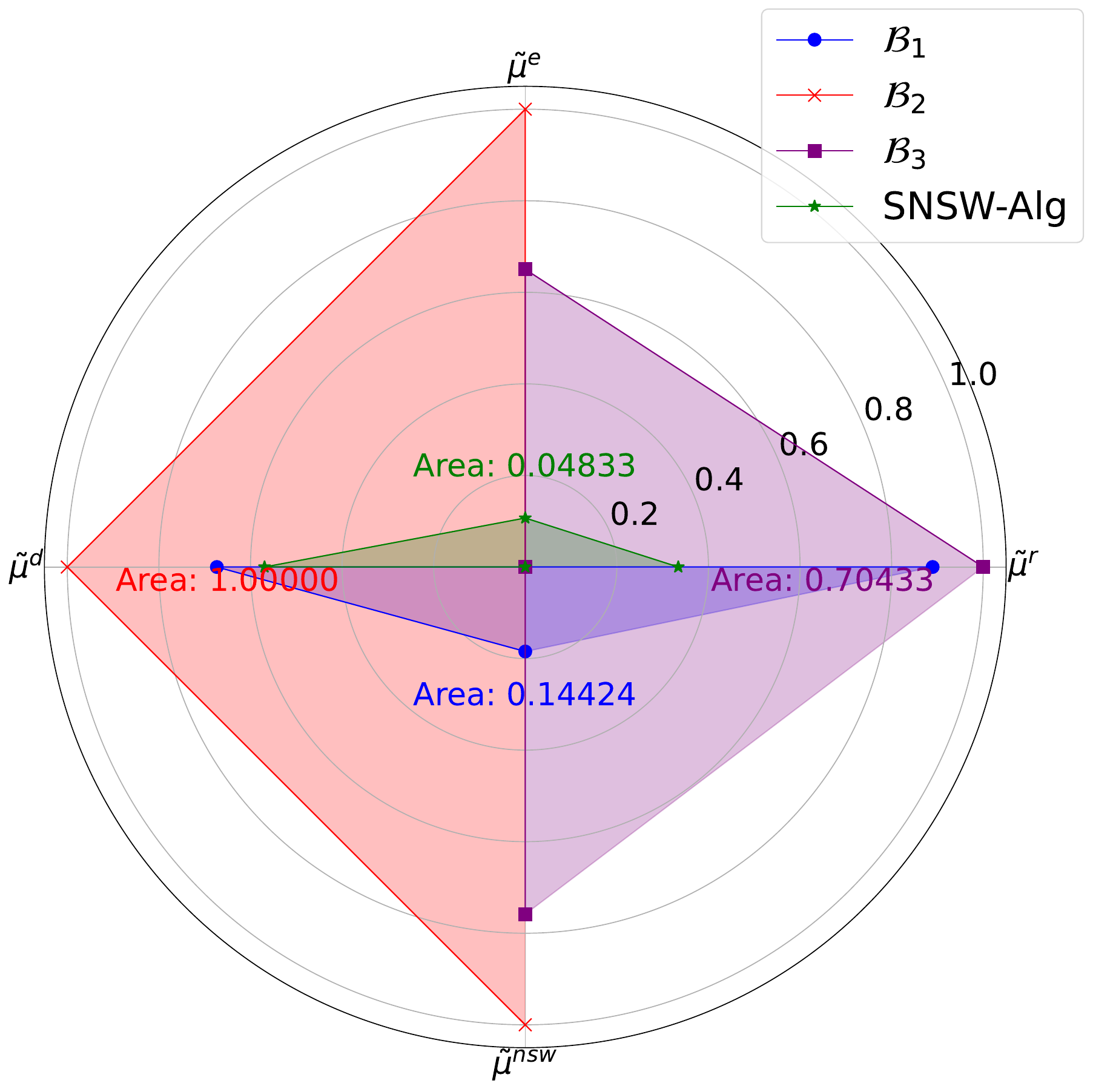}
        % \caption{Fourth figure}
        \label{fig:fig4_popularity_triangular}
    \end{subfigure}

    \caption{Circular Plot for Triangular Distribution}
    \label{fig:four_figures_popularity_triangular}
\end{figure}

%%%%%%%%%% Normal Plots %%%%%%%%
\begin{figure}[ht]
    \centering

    \begin{subfigure}{0.48\textwidth}
        \centering
        \includegraphics[width=\textwidth]{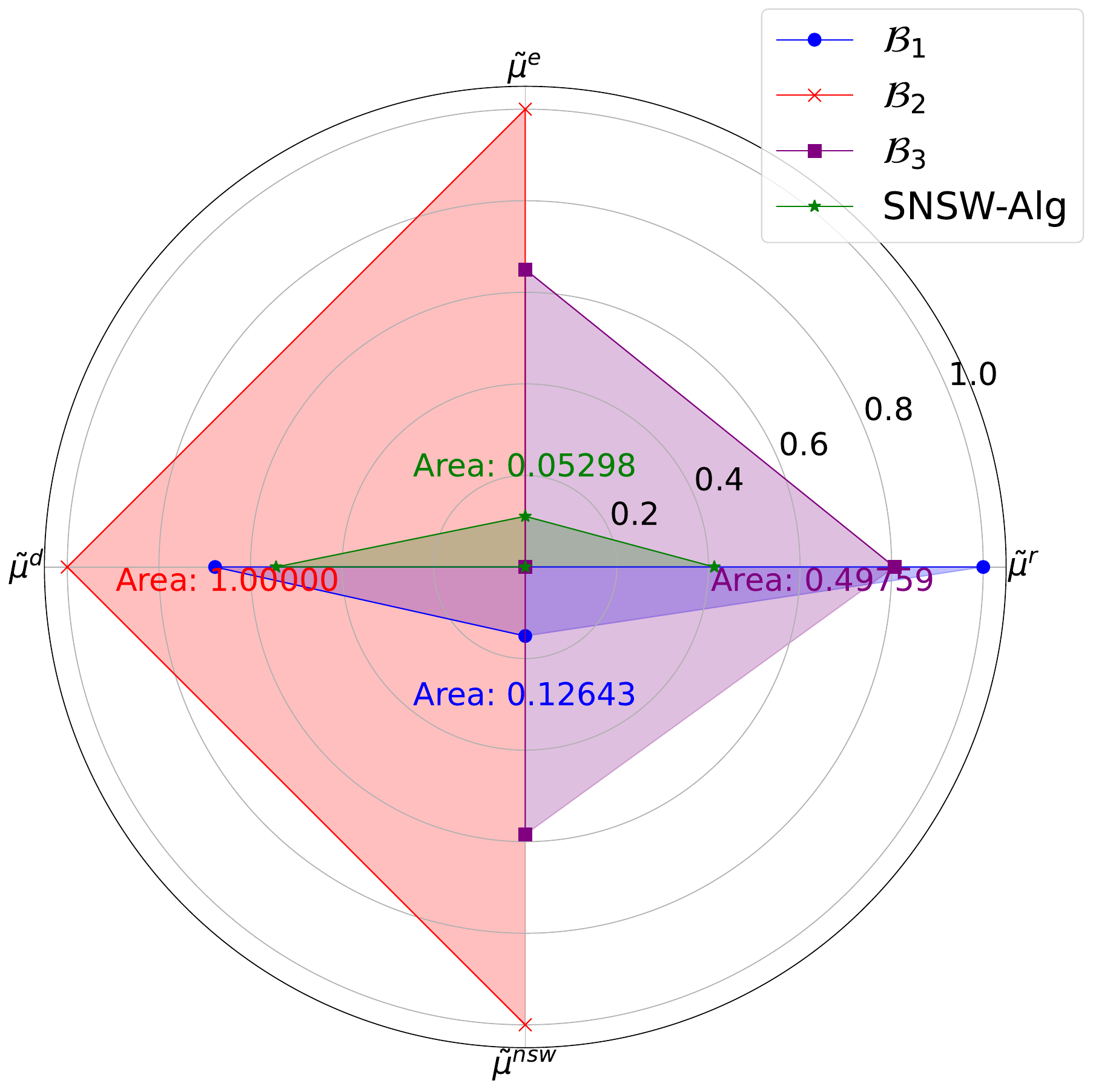}
        % \caption{First figure}
        \label{fig:fig1_popularity_normal}
    \end{subfigure}
    \hfill
    \begin{subfigure}{0.48\textwidth}
        \centering
        \includegraphics[width=\textwidth]{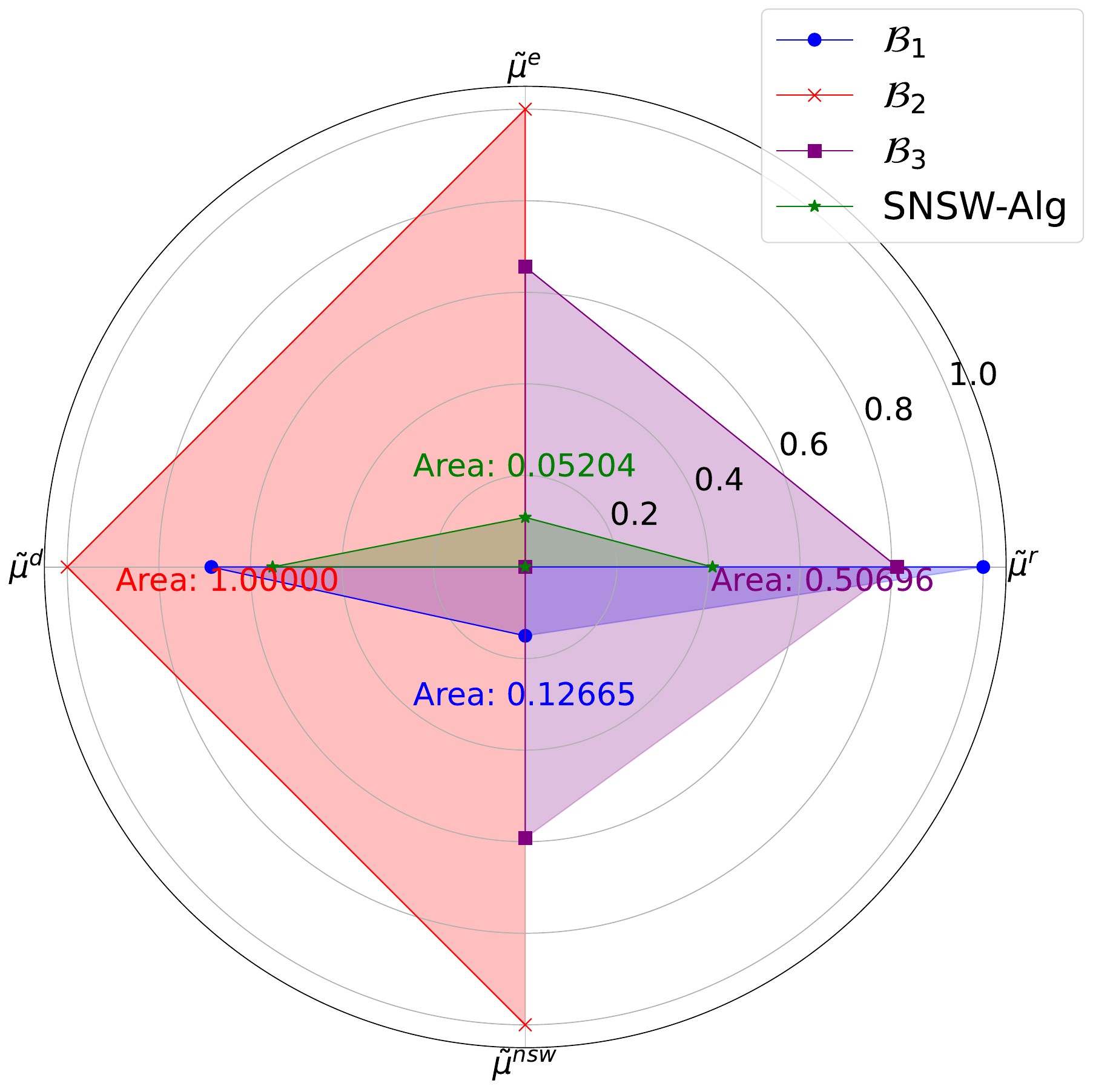}
        % \caption{Second figure}
        \label{fig:fig2_popularity_normal}
    \end{subfigure}

    \begin{subfigure}{0.48\textwidth}
        \centering
        \includegraphics[width=\textwidth]{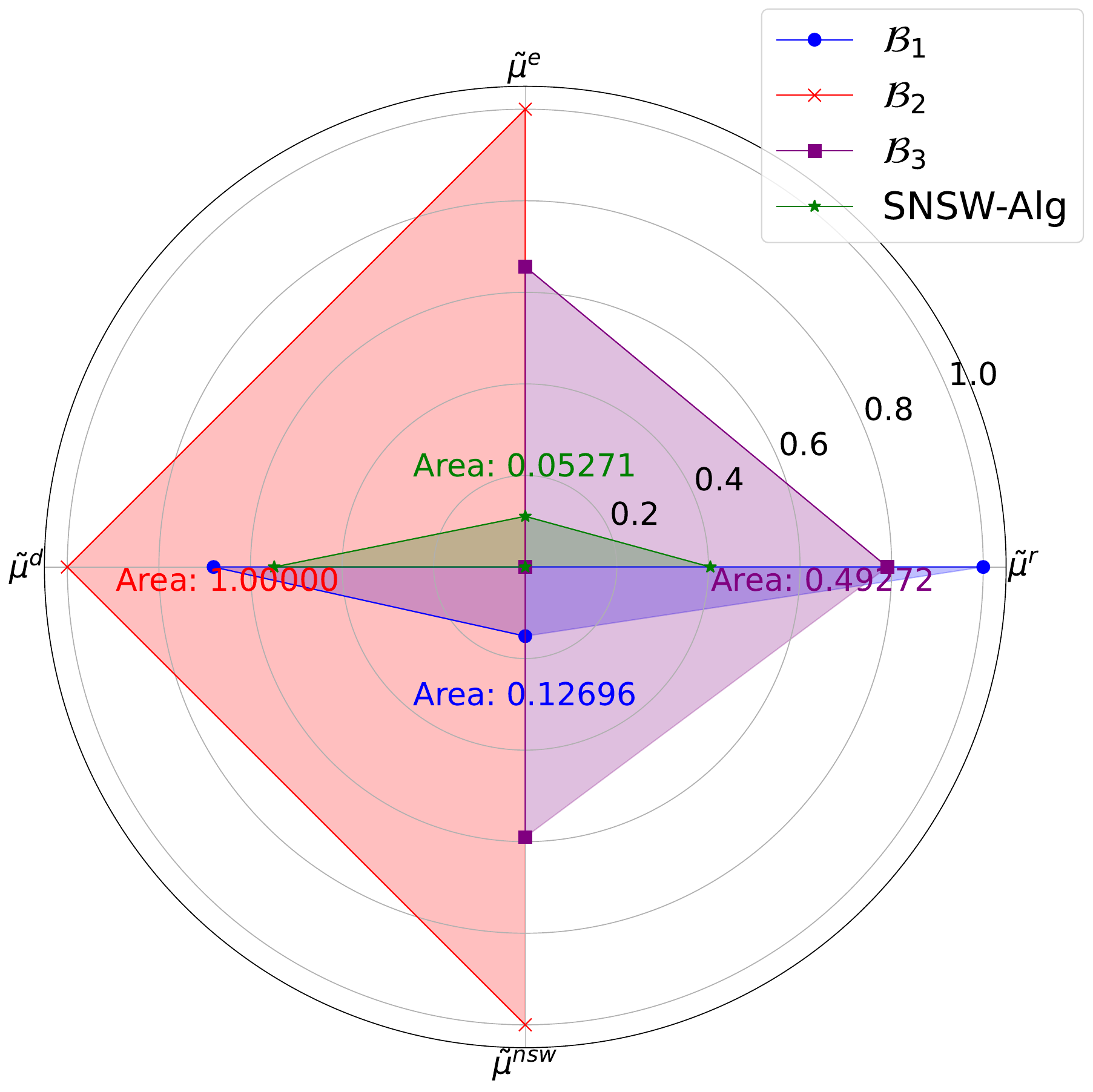}
        % \caption{Third figure}
        \label{fig:fig3_popularity_normal}
    \end{subfigure}
    \hfill
    \begin{subfigure}{0.48\textwidth}
        \centering
        \includegraphics[width=\textwidth]{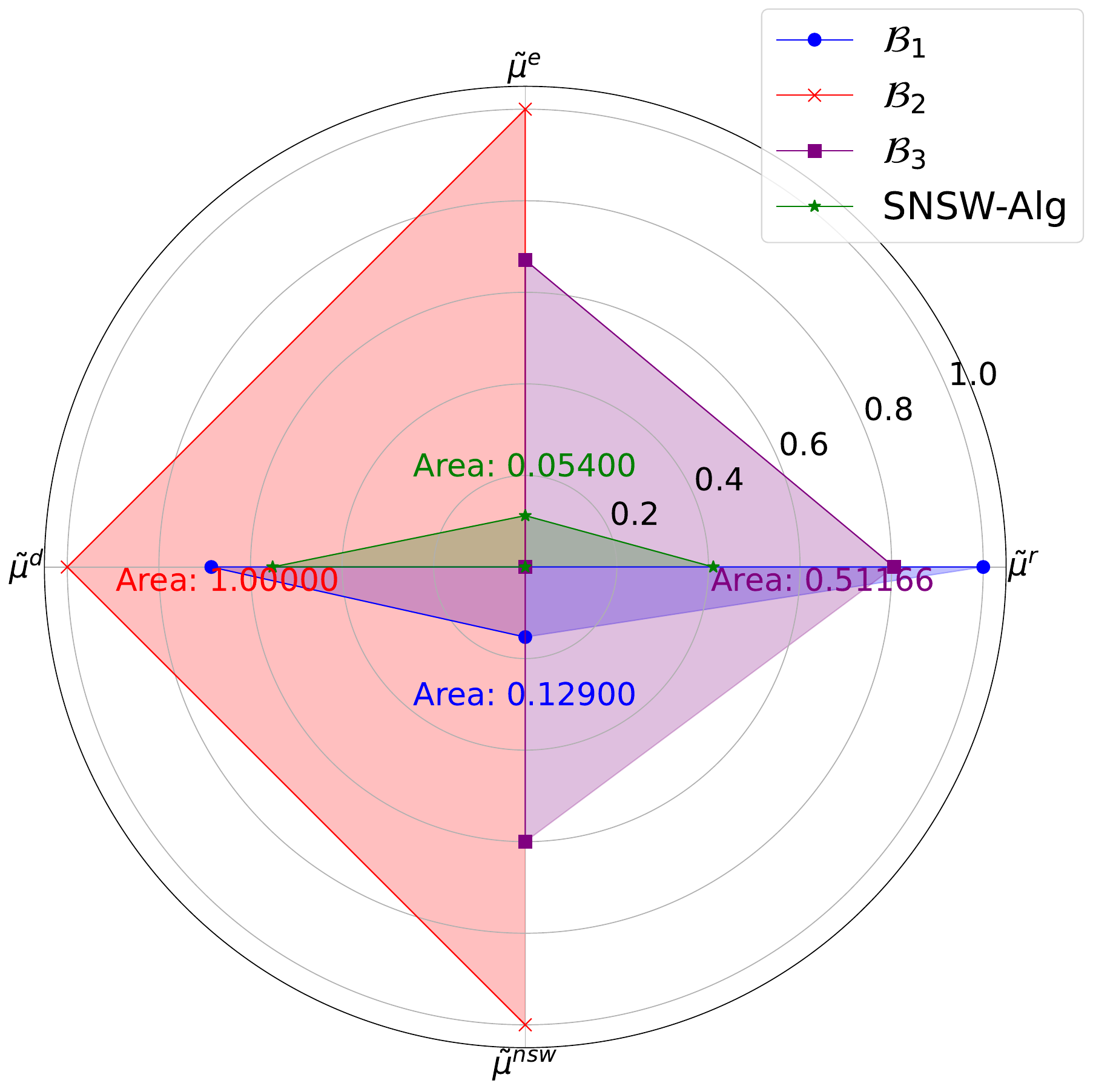}
        % \caption{Fourth figure}
        \label{fig:fig4_popularity_normal}
    \end{subfigure}

    \caption{Circular Plot for Normal Distribution}
    \label{fig:four_figures_popularity_normal}
\end{figure}

%%%%%%%%%% Uniform Default Plots %%%%%%%%
\begin{figure}[ht]
    \centering

    \begin{subfigure}{0.48\textwidth}
        \centering
        \includegraphics[width=\textwidth]{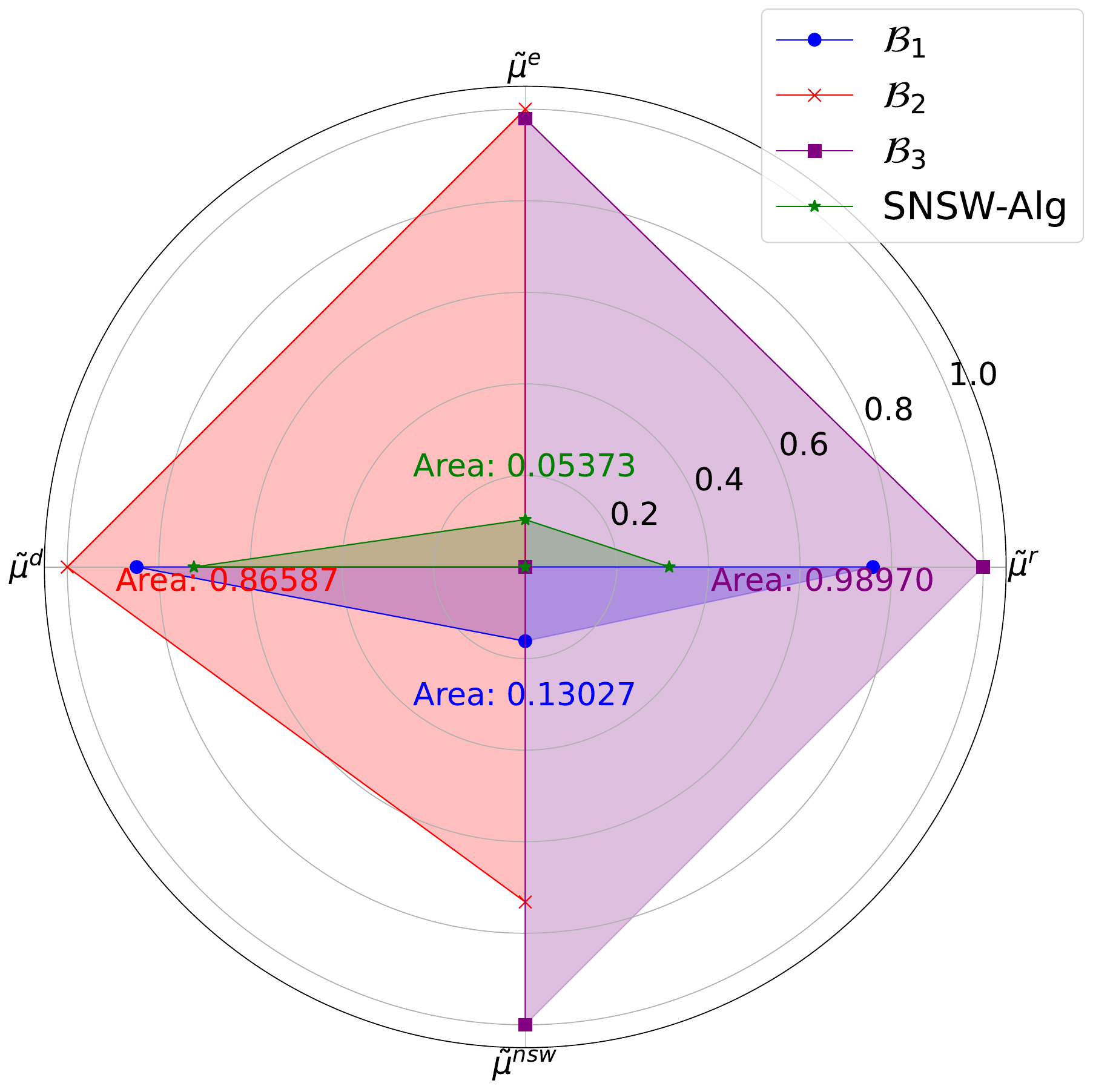}
        % \caption{First figure}
        \label{fig:fig1_default_uniform}
    \end{subfigure}
    \hfill
    \begin{subfigure}{0.48\textwidth}
        \centering
        \includegraphics[width=\textwidth]{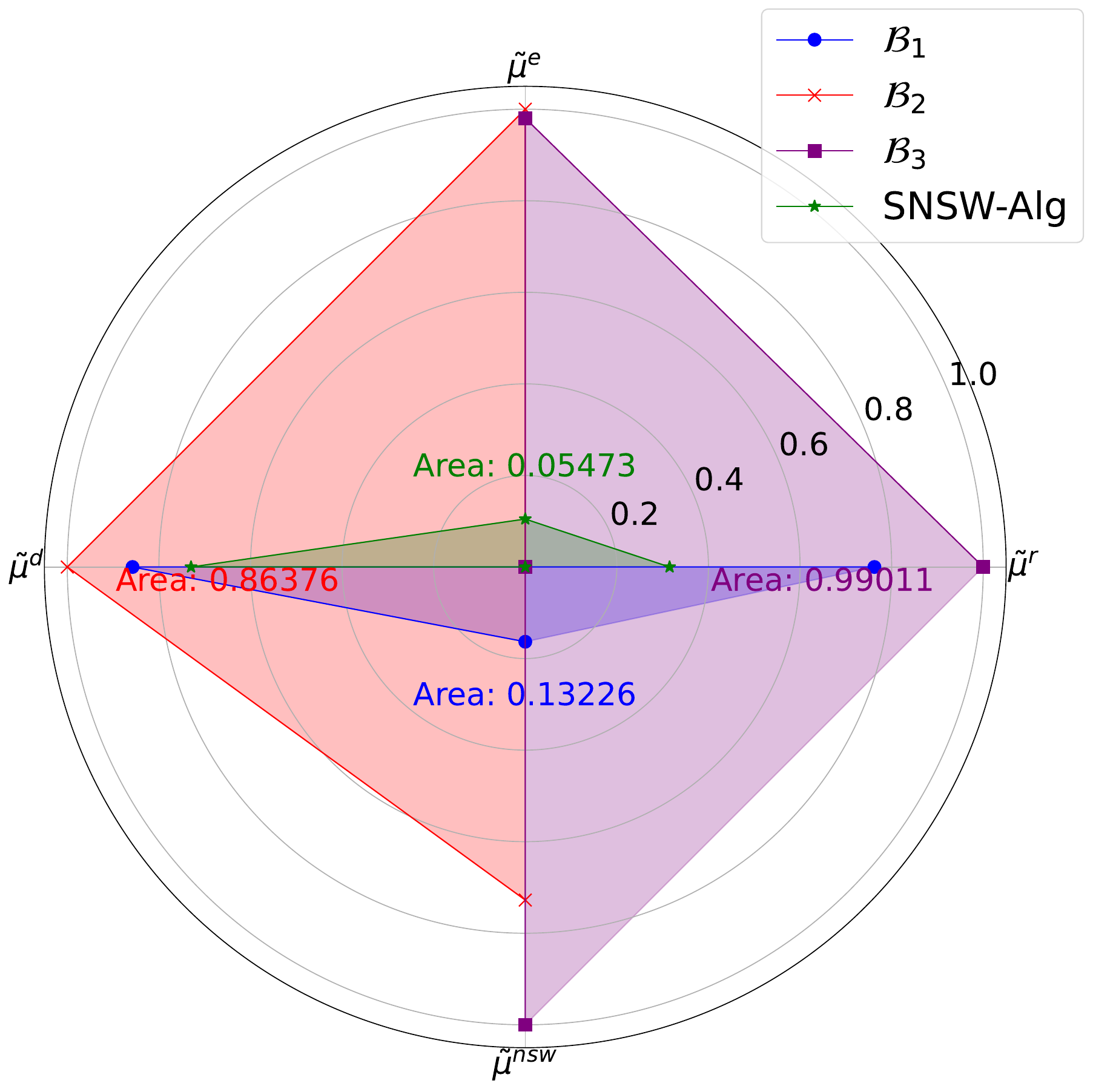}
        % \caption{Second figure}
        \label{fig:fig2_default_uniform}
    \end{subfigure}

    \begin{subfigure}{0.48\textwidth}
        \centering
        \includegraphics[width=\textwidth]{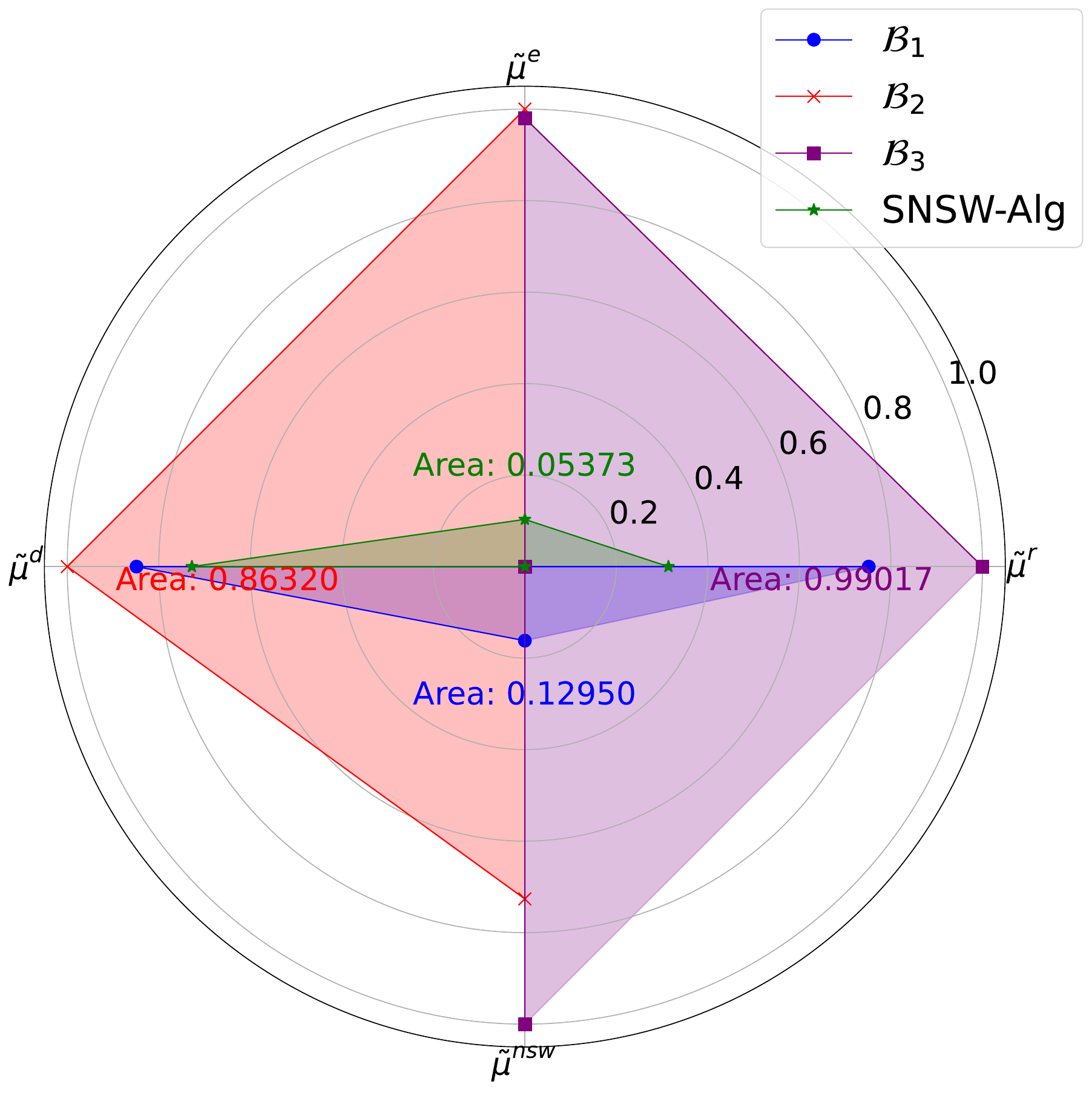}
        % \caption{Third figure}
        \label{fig:fig3_default_uniform}
    \end{subfigure}
    \hfill
    \begin{subfigure}{0.48\textwidth}
        \centering
        \includegraphics[width=\textwidth]{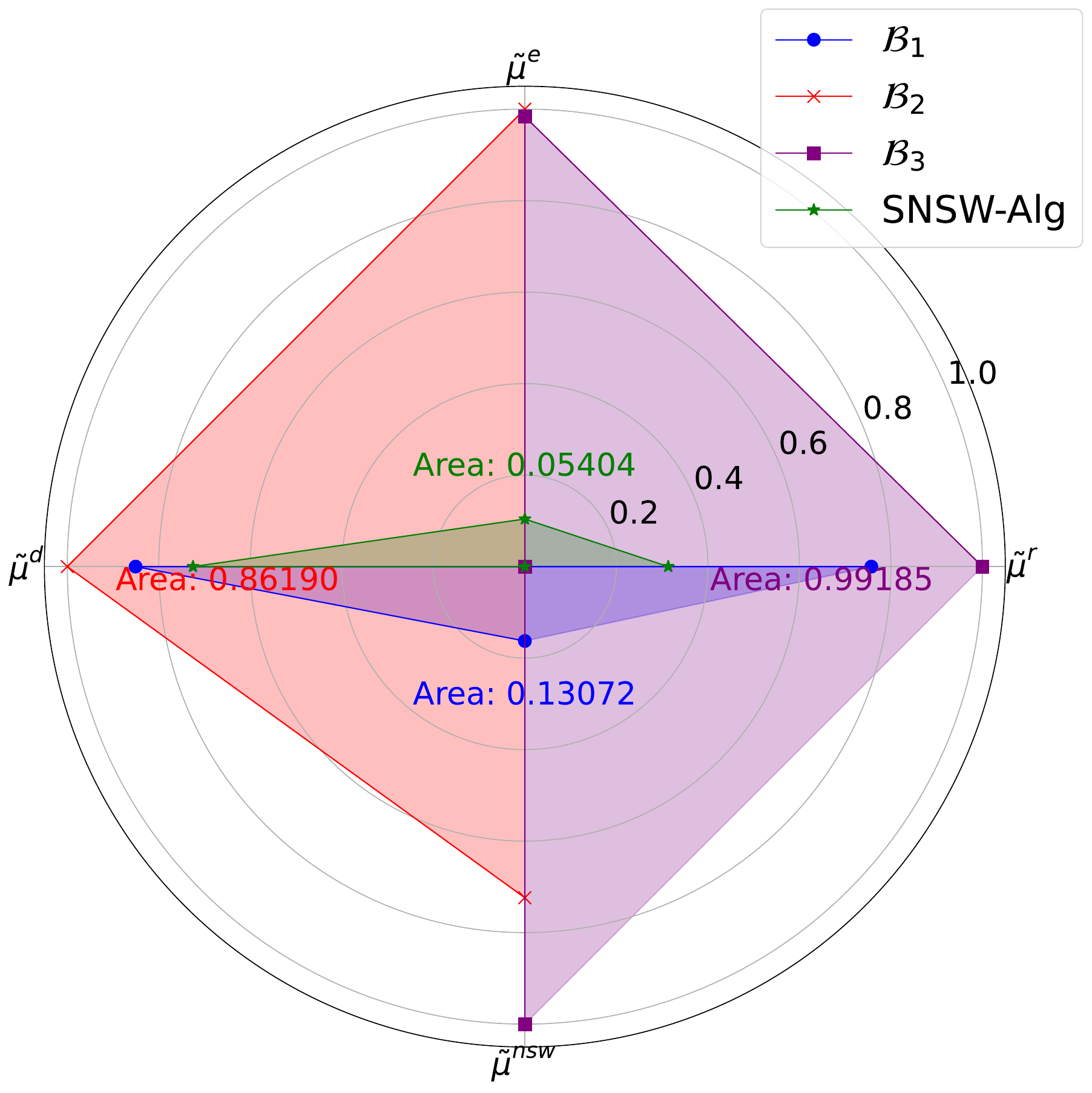}
        % \caption{Fourth figure}
        \label{fig:fig4_default_uniform}
    \end{subfigure}

    \caption{Circular Plot for Uniform Default Distribution}
    \label{fig:four_figures_default_uniform}
\end{figure}
\end{document}